%% file: d3.tex
\documentclass[conference]{IEEEtran}
\IEEEoverridecommandlockouts
\pdfoutput=1
\usepackage{cite}
\usepackage{amsmath,amsthm,amssymb,amsfonts}
\usepackage[ruled,linesnumbered]{algorithm2e}
\usepackage{graphicx}
\usepackage{graphbox}
\usepackage{textcomp}
\usepackage{xcolor}
\usepackage{centernot}
\usepackage{mathtools}
\usepackage{array}
\usepackage{float}
\usepackage{subcaption}
\usepackage{kantlipsum} %<- For dummy text
\usepackage{mwe} %<- For dummy images
\usepackage{xurl} 

\def\BibTeX{{\rm B\kern-.05em{\sc i\kern-.025em b}\kern-.08em
    T\kern-.1667em\lower.7ex\hbox{E}\kern-.125emX}}

\newcolumntype{L}[1]{>{\raggedright\arraybackslash}m{#1}}
\newcolumntype{C}[1]{>{\centering\arraybackslash}m{#1}}
\newcolumntype{R}[1]{>{\raggedleft\arraybackslash}m{#1}}

\newtheorem{proposition}{Proposition}

\begin{document}

\title{Dynamic DNN Decomposition for Lossless Synergistic Inference}

\author{
    \IEEEauthorblockN{Beibei Zhang, Tian Xiang, Hongxuan Zhang, Te Li, Shiqiang Zhu, Jianjun Gu}
    \IEEEauthorblockA{Intelligent Robotics Research Center, Zhejiang Lab, Hangzhou, China
    \\\{beibei, txiang, hongxuan, lite, zhusq, jgu\}@zhejianglab.com}
}

\maketitle
\input{abstract}
\input{introduction}
\input{related}
\input{methods}

\input{implementation}

\input{eval}

\input{conclusion}

\bibliographystyle{IEEEtran}
\bibliography{d3}

\end{document}

%% file: abstract.tex
\begin{abstract}
Deep neural networks (DNNs) sustain high performance in today's data processing applications. 
DNN inference is resource-intensive 
thus is difficult to fit into a mobile device. 
An alternative is to offload the DNN inference to a cloud server. 
However, such an approach requires heavy raw data transmission 
between the mobile device and the cloud server, 
which is not suitable for mission-critical 
and privacy-sensitive applications such as autopilot. 
To solve this problem, recent advances unleash DNN services using the edge computing paradigm. 
The existing approaches split a DNN into two parts 
and deploy the two partitions to computation nodes at two edge computing tiers. 
Nonetheless, these methods overlook collaborative device-edge-cloud computation resources. 
Besides, previous algorithms demand the whole DNN re-partitioning to adapt to 
computation resource changes and network dynamics. 
Moreover, for resource-demanding convolutional layers, 
prior works do not give a parallel processing strategy without loss of accuracy at the edge side. 
To tackle these issues, we propose D$^3$, a dynamic 
DNN decomposition system for synergistic inference without precision loss. 
The proposed system introduces a heuristic algorithm named 
horizontal partition algorithm to split a DNN into three parts. 
The algorithm can partially adjust the partitions at 
run time according to processing time and network conditions. 
At the edge side, a vertical separation module 
separates feature maps into tiles 
that can be independently run on different edge nodes in parallel. 
Extensive quantitative evaluation of five popular DNNs illustrates that D$^3$ 
outperforms the state-of-the-art counterparts up to 3.4$\times$ in end-to-end DNN inference time 
and reduces backbone network communication overhead up to 3.68$\times$.
\end{abstract}


%% file: introduction.tex
\section{introduction}
\begin{figure*}
	\centering
	\begin{subfigure}[t]{0.3\textwidth}
		\centering
		\includegraphics[width=\textwidth]{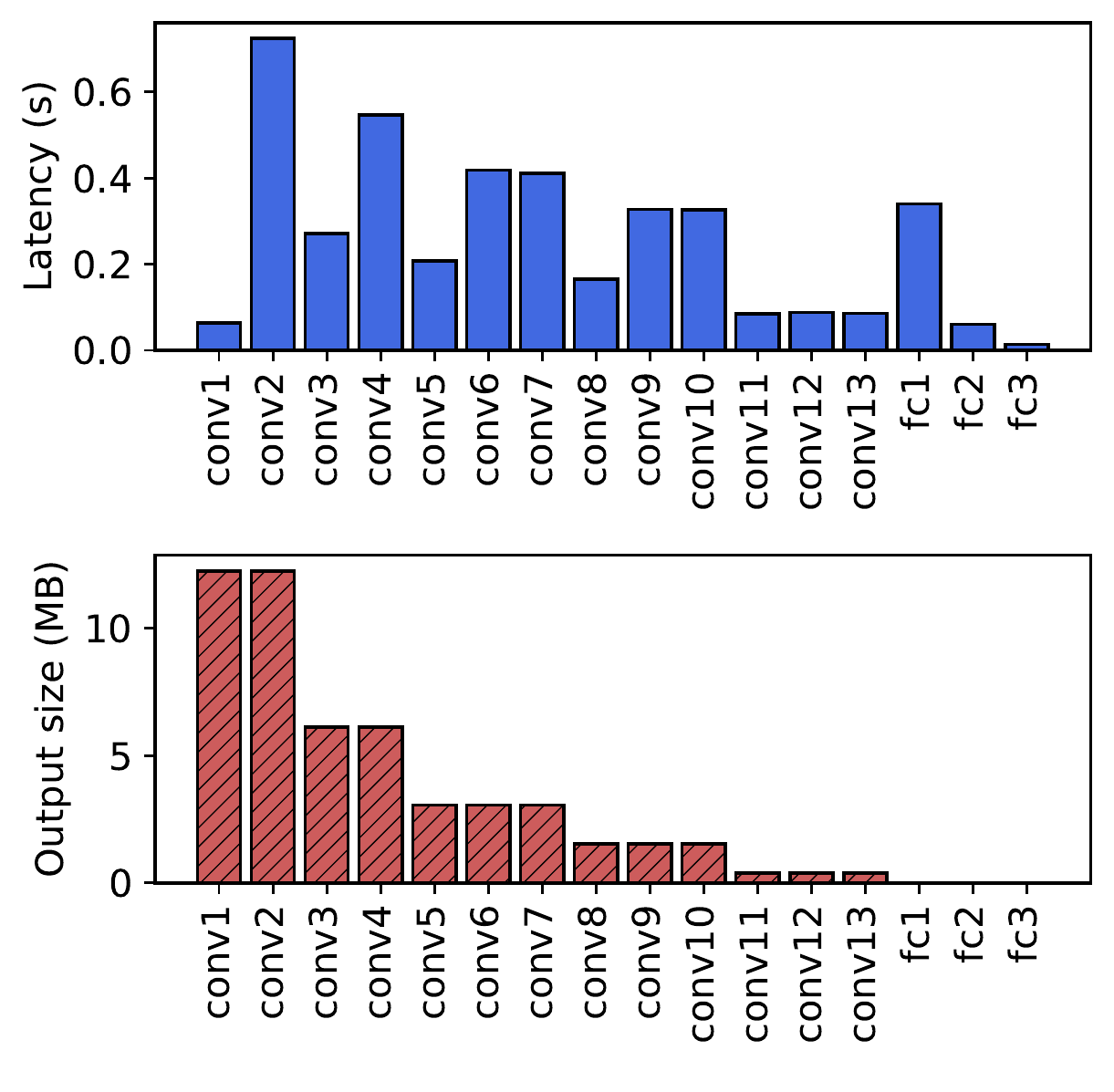}
		\caption{VGG-16}
		\label{fig:vgg16}
	\end{subfigure}
	\hfill
	\begin{subfigure}[t]{0.3\textwidth}
		\centering
		\includegraphics[width=\textwidth]{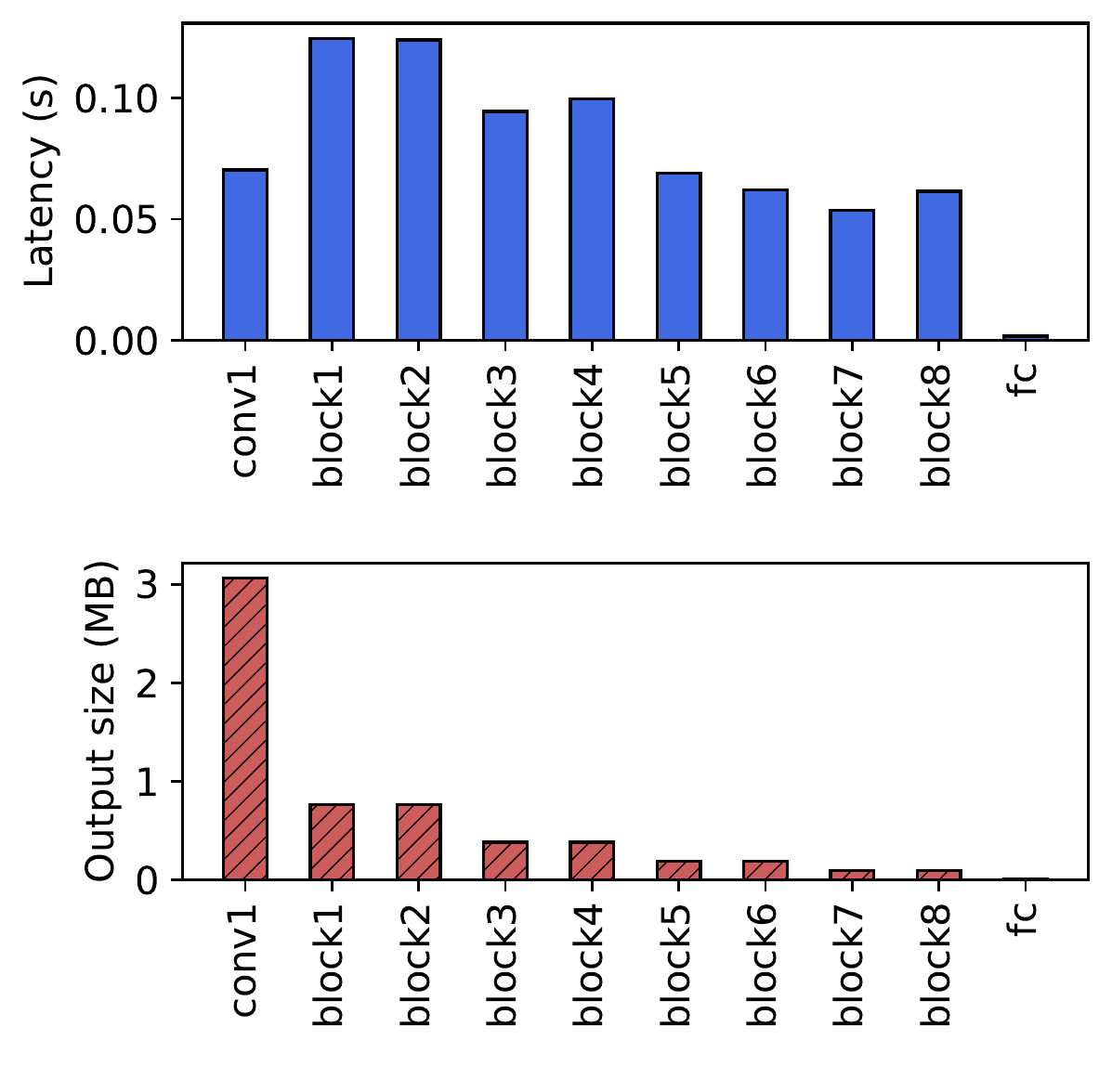}
		\caption{ResNet-18}
		\label{fig:resnet18}
	\end{subfigure}
	\hfill
	\begin{subfigure}[t]{0.3\textwidth}
		\centering
		\includegraphics[width=\textwidth]{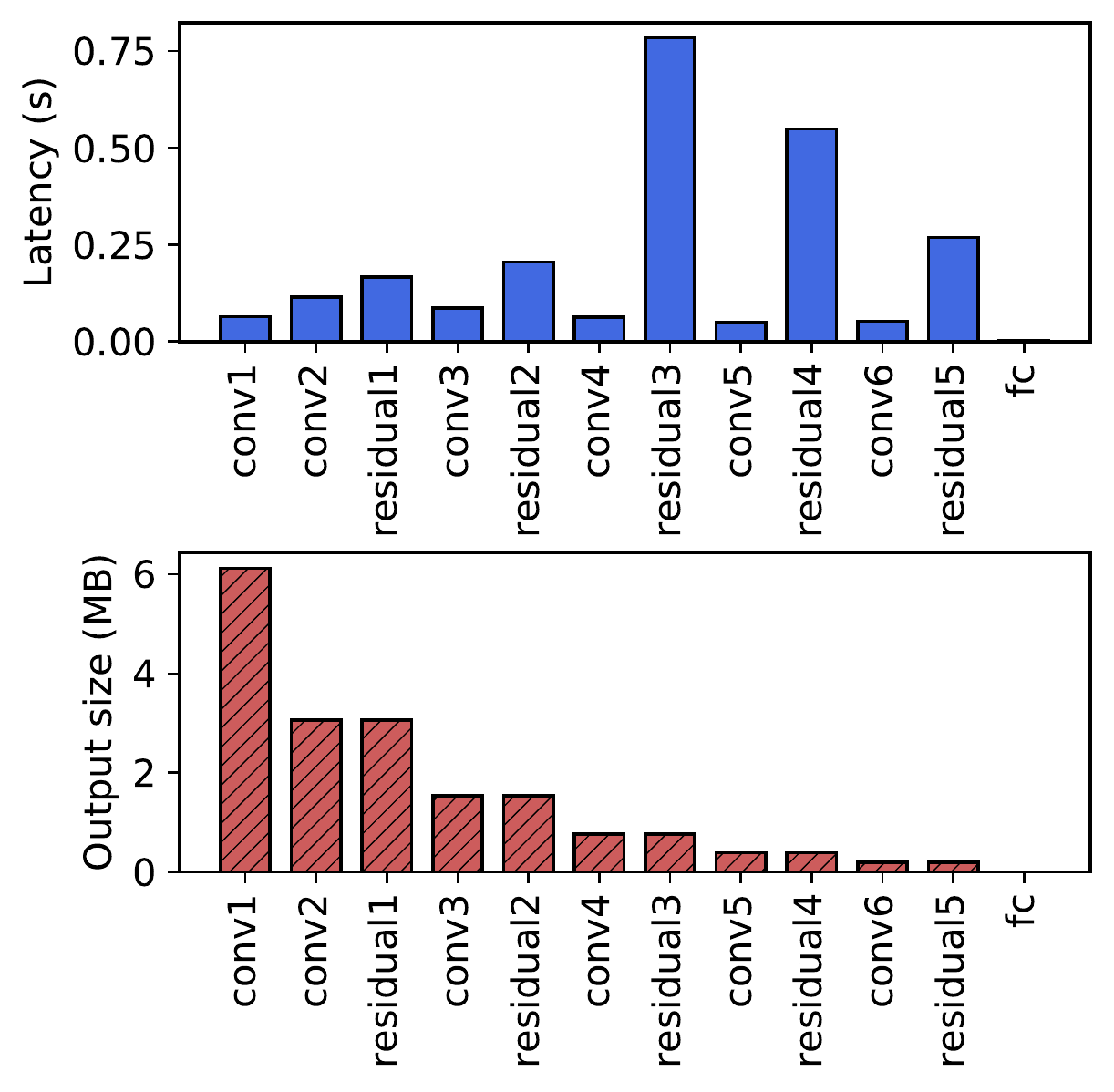}
		\caption{Darknet-53}
		\label{fig:darknet53}
	\end{subfigure}
	\caption{With the input size of $3 \times 224 \times 224$, we measure the layer-wise inference latency and the per-layer output size of VGG-16, ResNet-18, and Darknet-53 on a Raspberry Pi~4 model B running at 1.5GHz with 4 GB system memory. Each block or residual contains several convolutional layers.}
	\label{fig:infertime}
\end{figure*}

The proliferation of mobile devices such as smartphones and smart robotics 
brings a tremendous amount of data generated from users. 
On mobile devices, pervasive data processing applications 
including machine translation~\cite{bahdanau2014neural}, 
object detection~\cite{he2019tracking}, and many others 
process the data and share the data over a communication network. 
The results of these applications should be exceedingly accurate~\cite{he2016dual}. 
As such, deep neural networks (DNNs) become one of the de-facto solutions 
in the data processing applications due to its high accuracy. 
However, DNN inference requires abundant computation resources 
and consumes considerable energy~\cite{sze2020efficient}. 
Therefore, it is not suitable to deploy DNNs on mobile devices 
that have restricted computation power and limited energy. 

One of the popular approaches to tackle this issue is to 
offload the intensive DNN inference to a cloud server~\cite{wang2020convergence}. 
However, this approach involves 
transferring raw data collected by the data processing applications 
to the cloud server through a backbone network, 
which incurs long transmission delays and privacy concerns. 
With the observation that the intermediate result size of a DNN is 
significantly smaller than the raw data size, 
the idea of only transferring the intermediate result to the cloud server comes into existence. 
Recent research employs the \textit{edge computing paradigm} that comprises three computing tiers (i.e., device~\footnote{To avoid ambiguity, we use the term ``device'' to represent the ``device tier of edge computing paradigm'' and the term ``node'' to denote the ``computing device'' throughout this paper.}, edge, cloud) to solve the problem. 
Specifically, previous studies propose to split a DNN into two parts 
according to the processing time of DNN layers 
and the data transmission delay between two layers~\cite{teerapittayanon2017distributed, kang2017neurosurgeon, hu2019dynamic, zhou2019adaptive}. 
Generally, a mobile device collects raw data 
and passes the input to the first DNN partition located at an edge node. 
Next, the edge node processes the input and transmits the intermediate results to a cloud server. 
Finally, the cloud server handles the rest of the DNN inference 
that requires more processing capability. 
The collaborative computation leverages the resources provided by the edge and the cloud,  
reducing DNN inference latency and communication overhead over the network core. 

Nonetheless, the rapid development of hardware makes 
it possible to perform partial DNN inference on mobile devices~\cite{wang2020neural}. 
For instance, the latest smartphone has an octa-core CPU running up to 3.1 GHz 
and a GPU with~$1.37$ TFLOPS~\cite{snapdragon865plus}. 
In this context, current solutions fail to leverage the synergistic
device-edge-cloud computation power~\cite{hu2019dynamic}. 
Besides, prior works require re-partitioning the whole DNN 
to accommodate dynamics of computation resources and network bandwidth~\cite{zhou2019edge}. 
Moreover, in a joint DNN inference pipeline, 
the node with the most processing time becomes the bottleneck of the overall inference.
For an edge node with limited resources compared with a cloud node, 
if it is the bottleneck of the collaborative inference, 
the state-of-the-art does not provide a parallel processing strategy for 
convolutional layers assigned to the edge node without loss of accuracy~\cite{zhao2018deepthings}. 

To address these limitations, in this paper, we present a 
\textbf{d}ynamic \textbf{D}NN \textbf{d}ecomposition system named D$^3$ 
for lossless synergistic inference. 
D$^3$ accelerates DNN inference by leveraging the
synergistic device-edge-cloud computation power without precision loss.
In D$^3$, we employ a regression model that takes computation resources 
and DNN layer configurations as input and estimates the processing time of DNN layers. 
According to the per-layer execution time and the transmission delay between layers, 
a heuristic algorithm, which we refer to as \textbf{h}orizontal \textbf{p}artition \textbf{a}lgorithm (HPA),   
partitions a DNN into three parts, each of which runs on one computing tier. 
In the case of changes in the DNN layer processing time or the transmission delay, 
HPA can partially adjust the DNN segmentation, 
adapting to the changes via local updates. 
At a finer granularity, a \textbf{v}ertical \textbf{s}eparation \textbf{m}odule 
(VSM) further splits a stack of feature maps of convolutional layers into blocks spatially. 
We assign a block of correlated feature maps to an edge node. 
In this way, we can run feature maps of convolutional layers on the edge nodes 
in parallel independently, utilizing the edge resources. 

We implement D$^3$ and evaluate its performance 
based on ImageNet dataset~\cite{deng2009imagenet} using various real-world DNNs 
including AlexNet~\cite{krizhevsky2012imagenet}, VGG-16~\cite{simonyan2014very}, 
ResNet-18~\cite{he2016deep}, Darknet-53~\cite{redmon2018yolov3}, 
and Inception-v4~\cite{szegedy2016inception}. 
Our experimental results show that compared with the state-of-the-art counterparts, 
D$^3$ accelerates the DNN inference time up to 3.4$\times$ and 
reduces the communication overhead up to 3.68$\times$. 

The rest of this paper is organized as follows. 
First, we review related work in section~\ref{sec:related work}. 
Then, we describe the proposed system in section~\ref{sec:proposed system}. 
Section~\ref{sec:implementation} introduces the implementation and the test-bed of D$^3$. 
Extensive comparisons and evaluations between D$^3$ 
and its counterparts follow in section~\ref{sec:evaluation}. 
Finally, we conclude in section~\ref{sec:conclusion}.

%% file: related.tex
\section{related work}
\label{sec:related work}

\begin{figure*}[t]
	\centering
	\includegraphics[width=2\columnwidth]{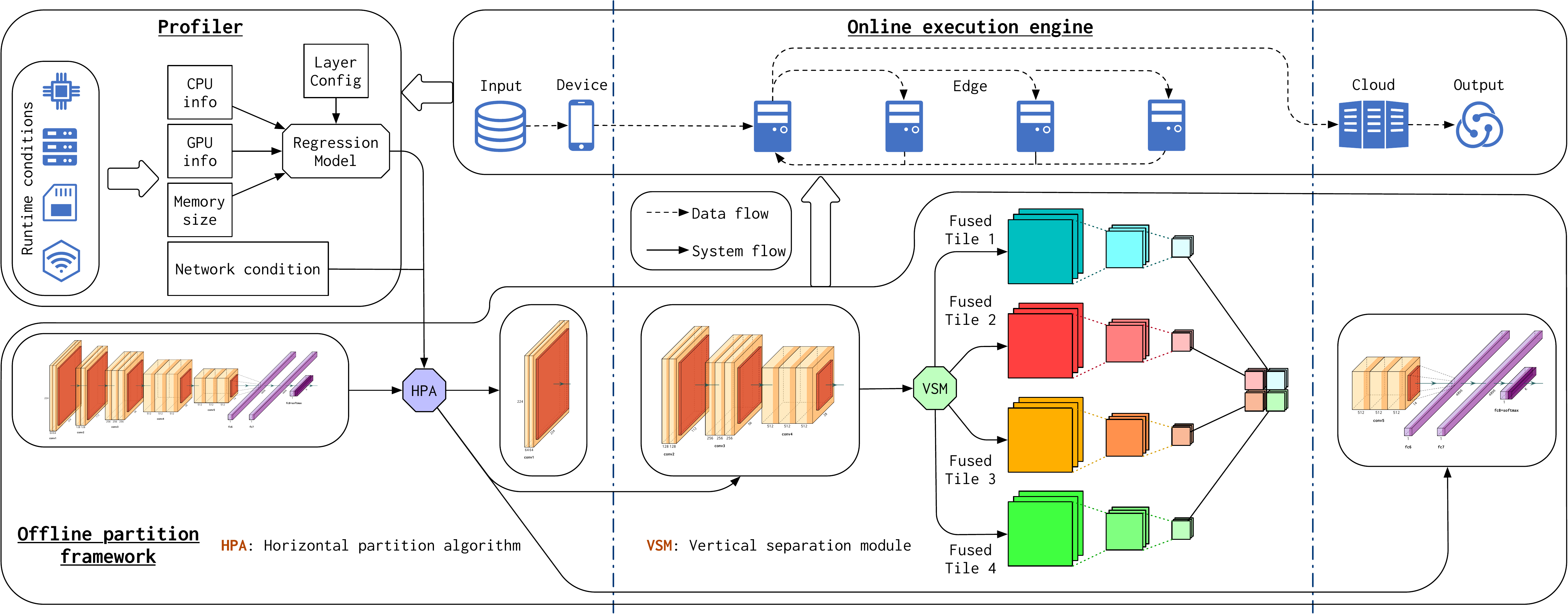}
	\caption{D$^3$ system architecture.}
	\label{fig:d4system}
\end{figure*}

Before delving into the D$^3$ system, 
we first provide an overview of the existing DNN inference acceleration mechanisms. 
The status-quo approaches that are related to our method can be classified into two categories. 
The first group aims to split a DNN model into partitions 
according to per-layer processing time and inter-layer transmission delay. 
The partitions are distributed to multiple computation nodes. 
We denote such approaches as horizontal partition. 
The second group, which we refer to as vertical separation, 
focuses on dividing the feature map of a convolutional layer spatially to multiple tiles 
hosted by multiple computation nodes. 

To accelerate DNN inference without sacrificing accuracy, 
a few previous works focus on the DNN horizontal partition.
Neurosurgeon~\cite{kang2017neurosurgeon} offloads computation 
from resource-constrained mobile devices to cloud servers. 
It splits a DNN of chain topology at a layer granularity
to minimize processing latency and energy consumption. 
IONN~\cite{jeong2018ionn} models a chain topology DNN as 
an auxiliary DAG and finds the optimal incremental offloading with the shortest path algorithm on the DAG. 
DADS~\cite{hu2019dynamic} extends the layer-wise partition 
to multi-branch DNNs represented by DAGs and exploits 
the min-cut algorithm to find the optimal cutting points.
It employs edge computing and deploys DNN layers to an edge node and a cloud server. 
However, DADS cannot generalize the min-cut approach to separate a DNN into more than two parts. 
DINA~\cite{mohammed2020distributed} presents 
an adaptive partition algorithm to divide DNN layers into pieces that can be smaller than a layer. 
It offloads the pieces to fog nodes based on a swap-matching algorithm. 
DDNN~\cite{teerapittayanon2017distributed} introduces the idea of early-exit 
that sacrifices accuracy in exchange for reducing inference delay and communication overhead. 
However, this requires a specific training process to ensure the prediction confidence. 
Thus, it is not suitable for accelerating the lossless inference of the trained DNN discussed in this paper. 
Edgent~\cite{li2019edge} combines the chain topology DNN segmentation and the early-exit 
method and performs collaborative device-edge DNN inference. 
SPINN~\cite{laskaridis2020spinn} applies the early-exit 
strategy to progressive inference and designs a scheduler to flexibly handle service-level agreements.

The convolutional layer, which is commonly deployed in current data processing applications, 
is one of the most resource-intensive components of a DNN~\cite{cao2019intelligent}.  
Fig.~\ref{fig:infertime} illustrates the per-layer inference latency 
and the inter-layer output size of three widely used DNNs. 
We notice that some convolutional layers require substantial computation resources. 
To explore inference parallelism for convolutional layers, 
prior works aim to separate feature maps spatially. 
MoDNN~\cite{mao2017modnn} proposes layer-wise parallelism 
that divides one feature map of a convolutional layer into pieces. 
Each computation node executes a part of the feature map and generates an output.
A host node gathers the output and re-partitions the feature map 
for parallel processing of the next convolutional layer. 
This procedure results in significant communication overhead.
DeepThings~\cite{zhao2018deepthings} removes the communication overhead 
by introducing a fused tile partition (FTP) 
that slices a stack of correlated feature maps spatially and distributes the stack to a computation node. 
AOFL~\cite{zhou2019adaptive} extends the idea of fused tiles 
and offers an algorithm to find the optimal tile partition according to 
resources of each computation node.

%% file: methods.tex
\section{proposed system}
\label{sec:proposed system}
This section formally describes D$^3$, a system that dynamically decomposes
a DNN to segments for collaborative inference over device, edge, and cloud with no precision loss.
We first discuss the edge computing framework, 
based on which our system is designed.
We then give an overview of D$^3$, followed by the modeling of our system. 
A regression model is provided to estimate the per-layer execution latency of a DNN. 
Next, we propose our horizontal partition algorithm that splits a DNN into three parts. 
Finally, we describe the vertical separation module that enables parallel convolutional layer inference.

\subsection{Edge Computing Paradigm}
The edge computing architecture comprises three tiers that are 
device, edge, and cloud. Each layer consists of multiple computation nodes.
From a computation perspective, 
edge nodes provide high computation capabilities compared with device nodes. 
Even so, as the edge nodes are often heterogeneous, 
the computation power of the edge nodes is still bounded compared with the cloud servers.
Consequently, we say that the computation resources are gradually 
increasing over device, edge, and cloud~\cite{abbas2017mobile}.
From a communication perspective, since we deploy edge nodes close to the data source, 
network bandwidth maintains high between device and edge. 
Nonetheless, as device nodes and edge nodes connect 
to cloud servers through a backbone network (e.g., the Internet backbone), 
the bandwidth to the cloud nodes remains limited~\cite{wu2019machine}. 
Compared with the transmission delay between the computing tiers, 
the in-memory transmission delay of a node or the transmission delay 
between two computation nodes within the same computing tier is negligible. 
To simplify the problem, without loss of generality, 
we assume that the transmission delay within each computing tier is infinitesimal. 

\subsection{System Overview}
We show the D$^3$ system architecture in Fig.~\ref{fig:d4system}. 
D$^3$ comprises a profiler, an offline partition framework, and an online execution engine. 
The profiler collects the operating conditions of 
computation nodes at device, edge, and cloud as well as the network status between tiers. 
A regression model takes the computation statistics as input and  
estimates the inference time of each DNN layer processed at different computing tiers. 
The offline partition framework comprises two components. 
A horizontal partition algorithm splits a DNN model into three parts 
according to the per-layer inference time and the transmission delay between layers. 
D$^3$ distributes the three DNN segments to three computation nodes located at 
device, edge, and cloud respectively. 
Considered that the computation resource of a single edge node is limited, 
if convolutional layers locate at an edge node, to further accelerate the inference, 
a vertical separation module divides a sequence of correlated feature maps 
to multiple feature map stacks, 
each of which is processed on an edge node independently. 
Fig.~\ref{fig:d4system} illustrates the condition 
that D$^3$ breaks a sequence of feature maps into four stacks.
The online execution engine orchestrates the distributed and parallelism processing, 
handling communication among partitions. 
For the online DNN inference, a computation node located at the device tier collects input 
and processes the input through its allocated DNN partition. 
The device node passes the output to an edge node, 
which splits the feature maps into multiple tiles 
and diffuses the tiles to other available edge nodes.  
After the processing, an edge node gathers and integrates the results,
and passes it to a cloud server for further execution. The cloud server produces the final output. 

\subsection{System Model}
\begin{figure}
	\centering
	\begin{subfigure}[t]{.22\textwidth}
		\centering
		\includegraphics[width=\textwidth]{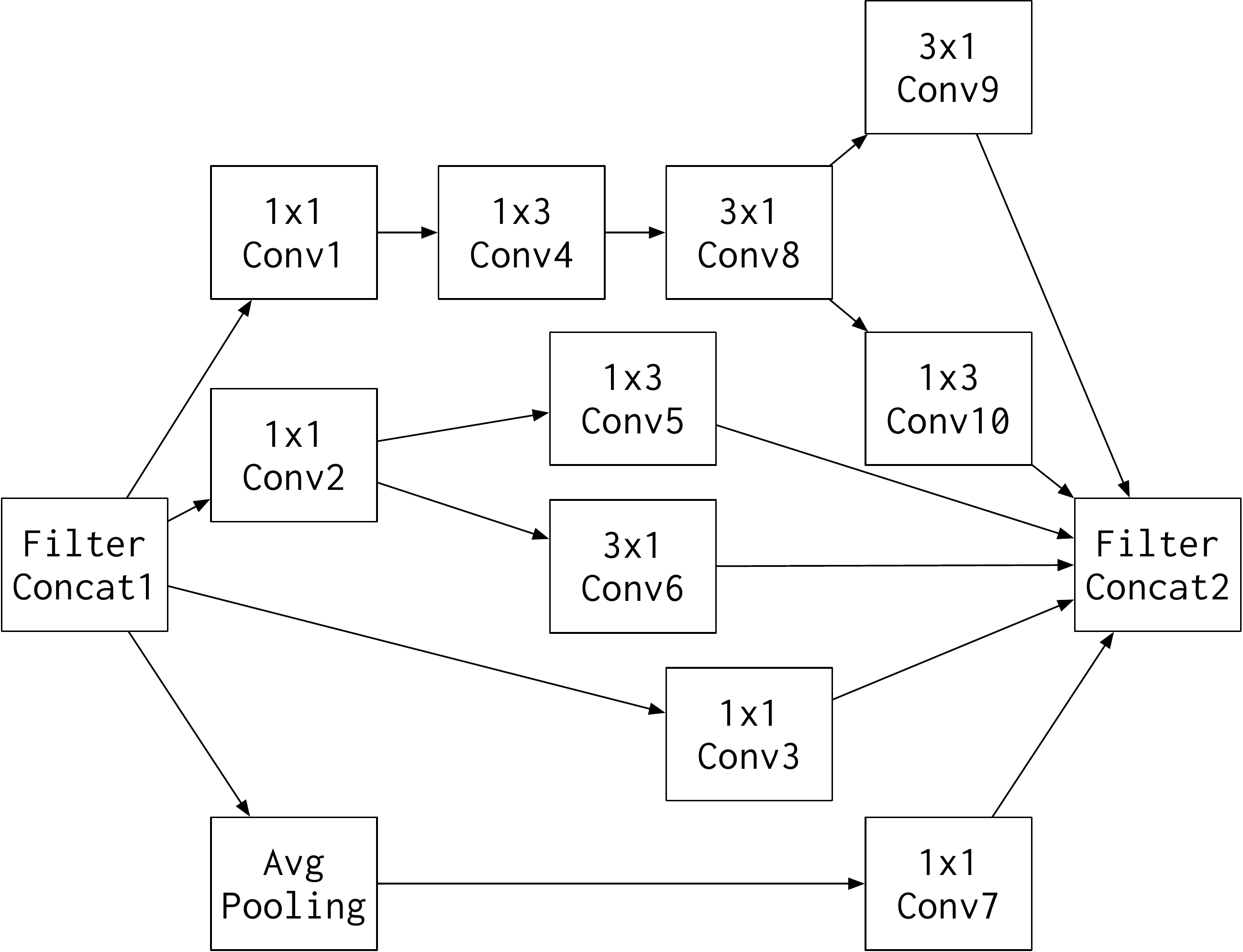}
		\caption{The grid module.}
		\label{fig:inception}
	\end{subfigure}
	\hfill
	\begin{subfigure}[t]{.245\textwidth}
		\centering
		\includegraphics[width=\textwidth]{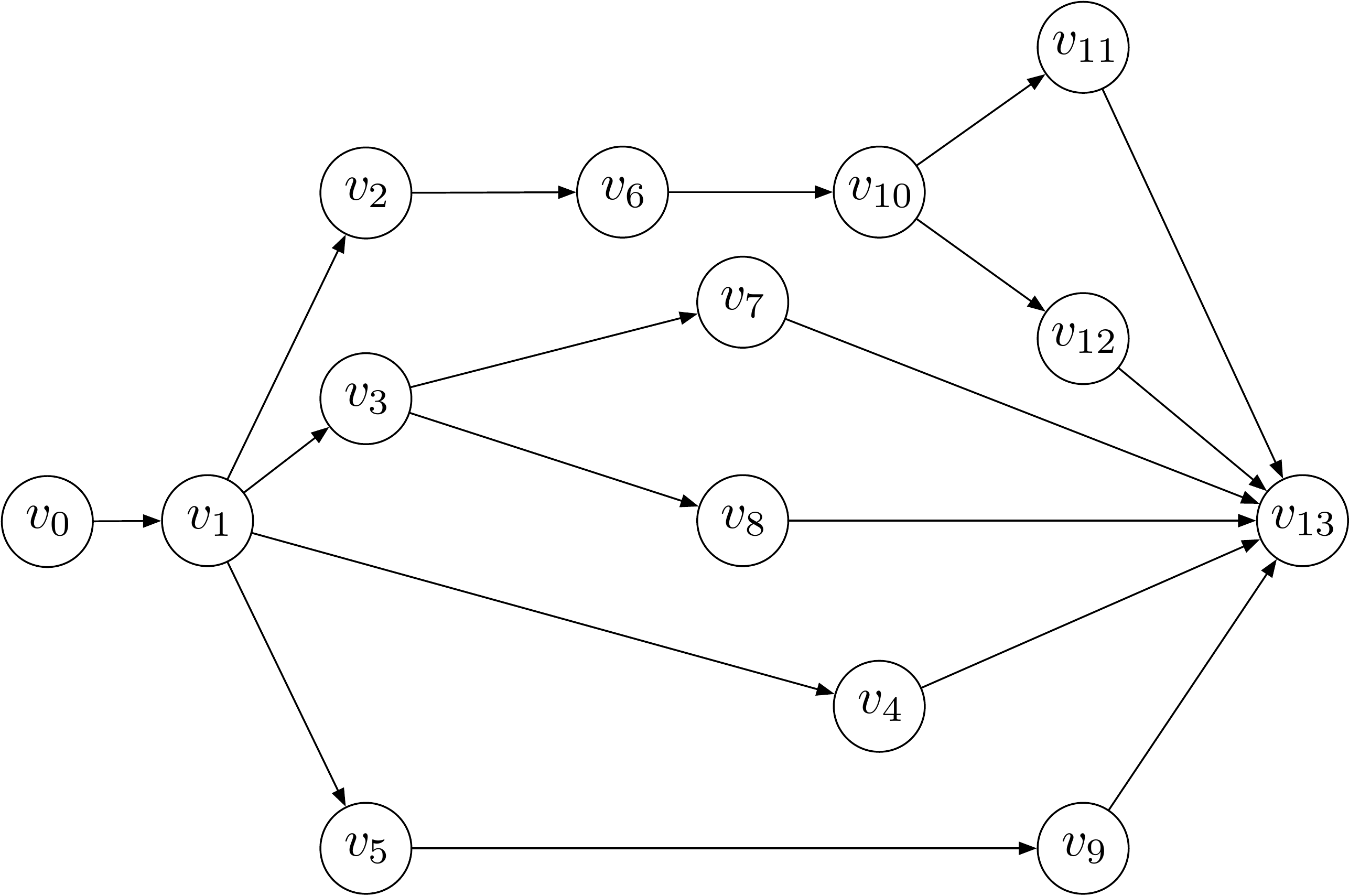}
		\caption{The DAG representation.}
		\label{fig:inceptiondag}
	\end{subfigure}
	\caption{The grid module of the Inception-v4 network and its DAG representation.}
	\label{fig:inceptionandinceptiondag}
\end{figure}

The smallest computation unit in a DNN model is a mathematical operator 
such as matrix multiplication and convolution. 
A DNN layer comprises one or multiple mathematical operators.
DNN layers constitute a computation graph that describes the DNN inference process. 
We describe that the computation graph can be modeled as a Directed Acyclic Graph (DAG).

Given a DNN model, we refer to the DNN layers in the model 
as a set of vertices $\{ v_1, v_2, \dots, v_n\}$ in a graph 
where $n$ is the number of layers in the model 
and $v_i$ corresponds to the $i$-th DNN layer in the model. 
To facilitate our algorithm, we introduce a virtual input vertex in the graph 
to indicate the starting point of a DNN and represent it using $v_0$.
For any two vertices $v_i$ and $v_j$ in the graph, 
we introduce a directed link~$(v_i, v_j)$ in the graph if and only if layer~$i$ is computed 
before layer~$j$ and the output of layer~$i$ serves as the input of layer~$j$. 
The graph is a DAG.

With above settings, we denote the given DNN model by the following DAG 
\begin{equation}
\label{eqn:dag}
	\mathcal{G} = (\mathcal{V}, \mathcal{L})
\end{equation}
where $\mathcal V = \{v_0, v_1, \dots, v_n\}$ and $\mathcal L\subset \mathcal V\times \mathcal V$ is the set of directed links in the DAG. 
As an example, Inception-v4 network~\cite{szegedy2016inception} 
is a multi-branch DNN that is depicted as a DAG. 
Fig.~\ref{fig:inception} shows the grid module of the 
Inception-v4 network and Fig.~\ref{fig:inceptiondag} 
illustrates its DAG representation. 

Next, we define relations between two vertices in $\mathcal{G}$. 
If there is a directed link from $v_i$ to $v_j$, 
then $v_i$ is a \textit{direct predecessor} of $v_j$ 
and $v_j$ is a \textit{direct successor} of $v_i$.  
For $v_i \in \mathcal{V}$, we refer to the set of its direct predecessors 
as $\mathcal{V}_i^p$ where $\mathcal{V}_i^p \subseteq \mathcal{V}$. 

The processing time of a DNN layer varies when the layer 
locates at different computing tiers. 
We use $d, e, c$ to express device, edge, and cloud tier respectively. 
For $v_i \in \mathcal{V}$, let $l_{i} \in \{d,e,c\}$ indicate 
the tier where vertex~$v_i$ is processed. 
Considered that the data of a DNN model flows from a device node, across an edge node, to a cloud node, 
to assist our algorithm, we define an order $d \succ e \succ c$, 
from which the $l_i$ of $v_i$ is selected.
We employ $t_i^d, t_i^e$, and $t_i^c$ to signify the processing time of 
$v_i$ at $d, e, c$ accordingly. Typically, we have $t_i^d > t_i^e > t_i^c$.  
$\forall v_i \in \mathcal{V}$, we assign $\mathcal{T}_{v_i} = \{t_i^d, t_i^e, t_i^c\}$ 
as vertex weight of $v_i$. 
For a directed link~$(v_i, v_j)$, if $v_i$ and $v_j$ are at distinct computing tiers, 
there is a transmission delay for data output from $v_i$ to $v_j$. 
We refer to the transmission delay from $v_i$ to $v_j$ over device and edge, 
edge and cloud, device and cloud 
as $t_{ij}^{[d, e]}$, $t_{ij}^{[e, c]}$, and $t_{ij}^{[d, c]}$ respectively. 
We assume that the two-way transmission delays between the two tiers are the same, which means 
$t_{ij}^{[d, e]}=t_{ij}^{[e, d]}$, $t_{ij}^{[e, c]}=t_{ij}^{[c, e]}$, 
and $t_{ij}^{[d, c]}=t_{ij}^{[c, d]}$. 
Generally, we have $t_{ij}^{[d, e]} < t_{ij}^{[e, c]} \leq t_{ij}^{[d, c]}$. 
For the condition that $v_i$ and $v_j$ locate at the same tier, 
the transmission delay is approximately~$0$. 
Therefore, $\forall (v_i, v_j) \in \mathcal{L}$,  
we allocate weight $\mathcal{T}_{(v_i, v_j)} = \{t_{ij}^{[d, e]}, t_{ij}^{[e, c]}, t_{ij}^{[d, c]}, 0\}$ 
to the directed link~$(v_i, v_j)$.

\subsection{Latency Estimation}
\begin{figure}[t]
	\centering
	\begin{subfigure}[t]{0.24\textwidth}
		\centering
		\includegraphics[width=\textwidth]{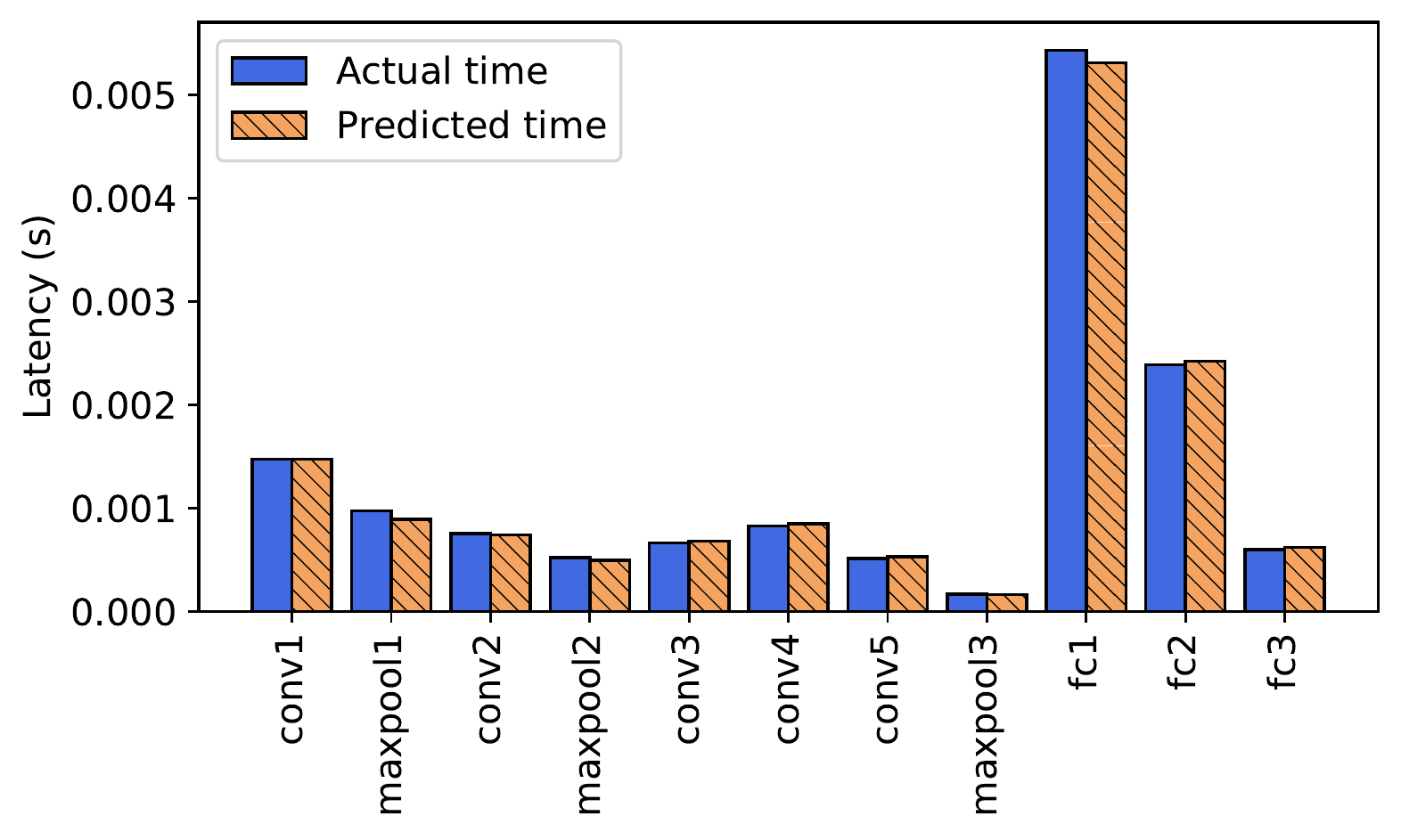}
		\caption{CPU}
		\label{fig:regressionvgg16}
	\end{subfigure}
	\hfill
	\begin{subfigure}[t]{0.24\textwidth}
		\centering
		\includegraphics[width=\textwidth]{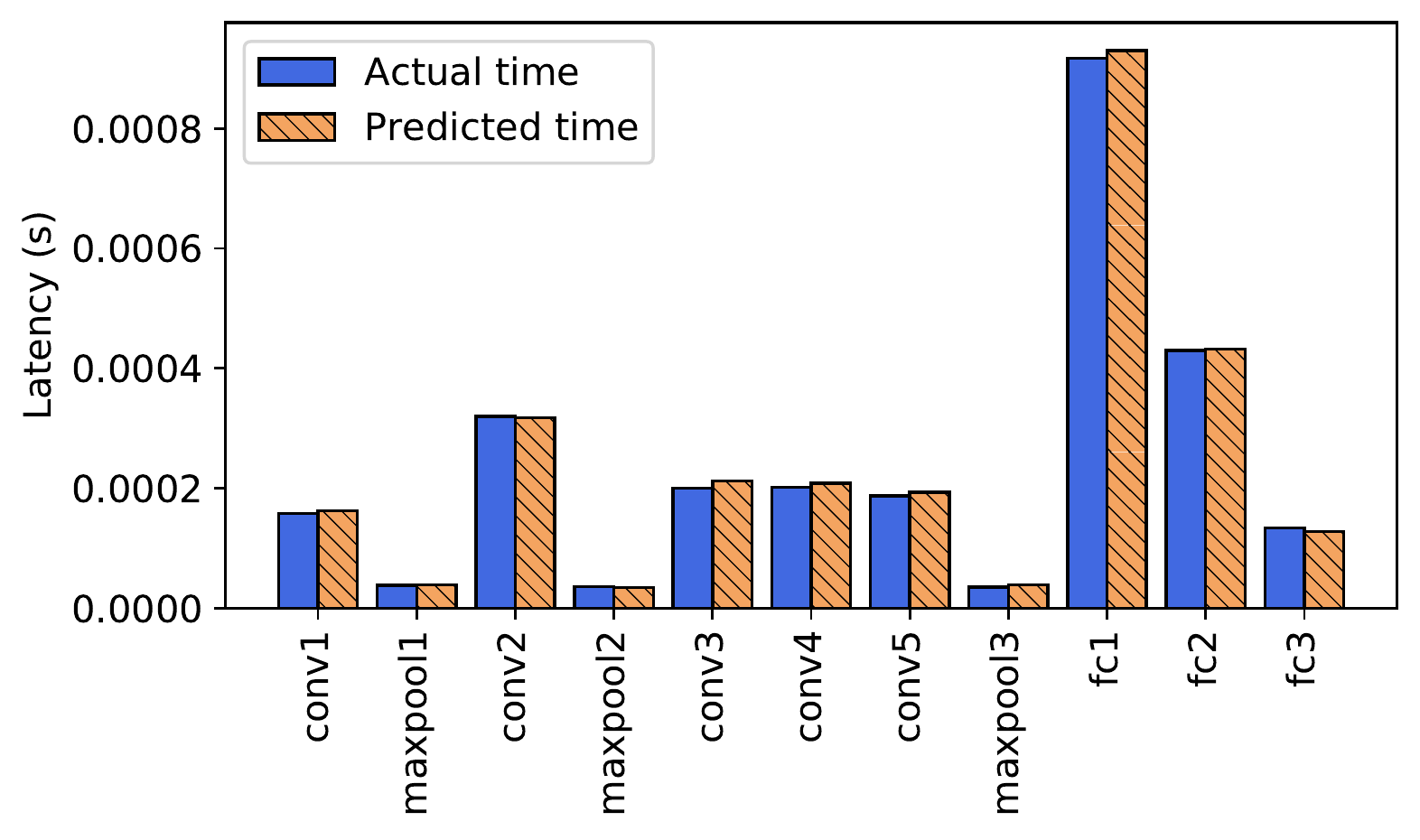}
		\caption{GPU}
		\label{fig:regressionresnet18}
	\end{subfigure}
	\caption{With the input size of $3 \times 224 \times 224$, the per-layer actual processing time and the predicted processing time of AlexNet on Intel Core i7-8700 CPU~\cite{intel8700cpu} and NVIDIA GeForce RTX 2080 Ti GPU~\cite{nvidia2080gpu}.}
	\label{fig:regressiontime}
\end{figure}

To determine the vertex weight $\mathcal{T}_{v_i}$ for all $v_i \in \mathcal{V}$ , 
one of the methods is to process DNN layers on the spot. 
However, executing every DNN layer at device, 
edge, and cloud nodes is impractical and time-consuming. 
Moreover, the computation resources of the nodes vary timely, 
thus $t_i^{l_i}$ of the actual measurement is inaccurate 
to resolve the optimal tier assignment of every DNN layer. 
Hence, such an approach is not feasible. 
To solve this problem, we employ a regression model 
that considers computation resources and DNN layer configurations to 
estimate the processing time of each DNN layer. 
We refer to the computation resources as the computation capabilities defined by CPU, GPU, 
and memory size. The DNN layer configurations include DNN layer types (e.g., convolution, ReLU, etc.) 
and DNN layer hyper-parameters (e.g., stride, input size, etc.). 
In Fig.~\ref{fig:regressiontime}, our regression model demonstrates that the actual processing time 
and the predicted processing time of each AlexNet~\cite{krizhevsky2012imagenet} layer are similar. 
We compute the link weight $\mathcal{T}_{(v_i, v_j)}$ between $v_i$ and $v_j$ 
by using the output data size of $v_i$ divided by
the network bandwidth between $l_i$ and $l_j$, which is monitored by the profiler. 
Note that if $l_i=l_j$, the transmission latency between $v_i$ and $v_j$ is $0$.

\subsection{Horizontal Partition Algorithm}

\begin{figure}[t]
	\centering
	\begin{subfigure}[t]{0.24\textwidth}
		\centering
		\includegraphics[width=\columnwidth]{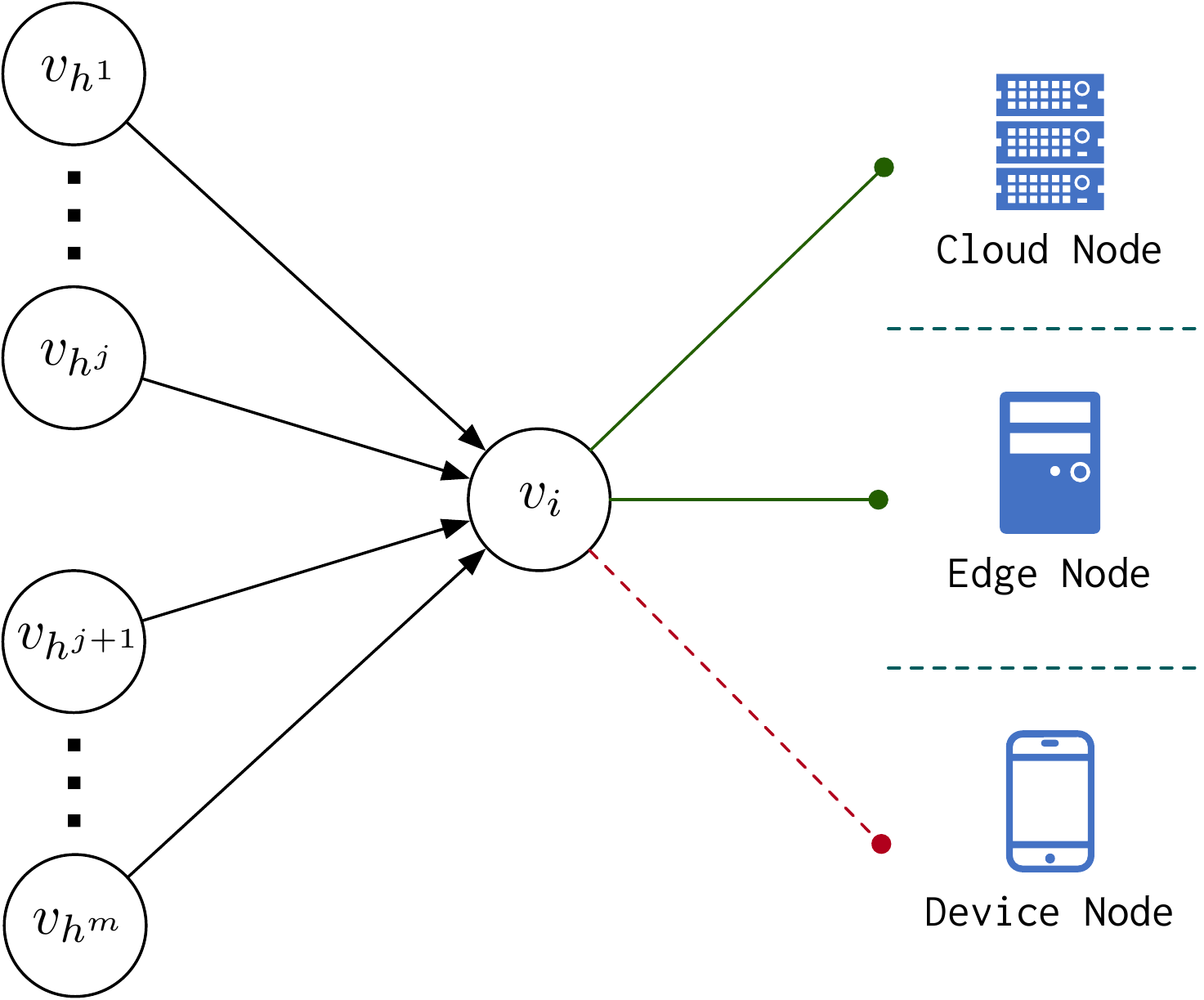}
		\caption{}
		\label{fig:proof1}
	\end{subfigure}
	\hfill
	\begin{subfigure}[t]{0.24\textwidth}
		\centering
		\includegraphics[width=\columnwidth]{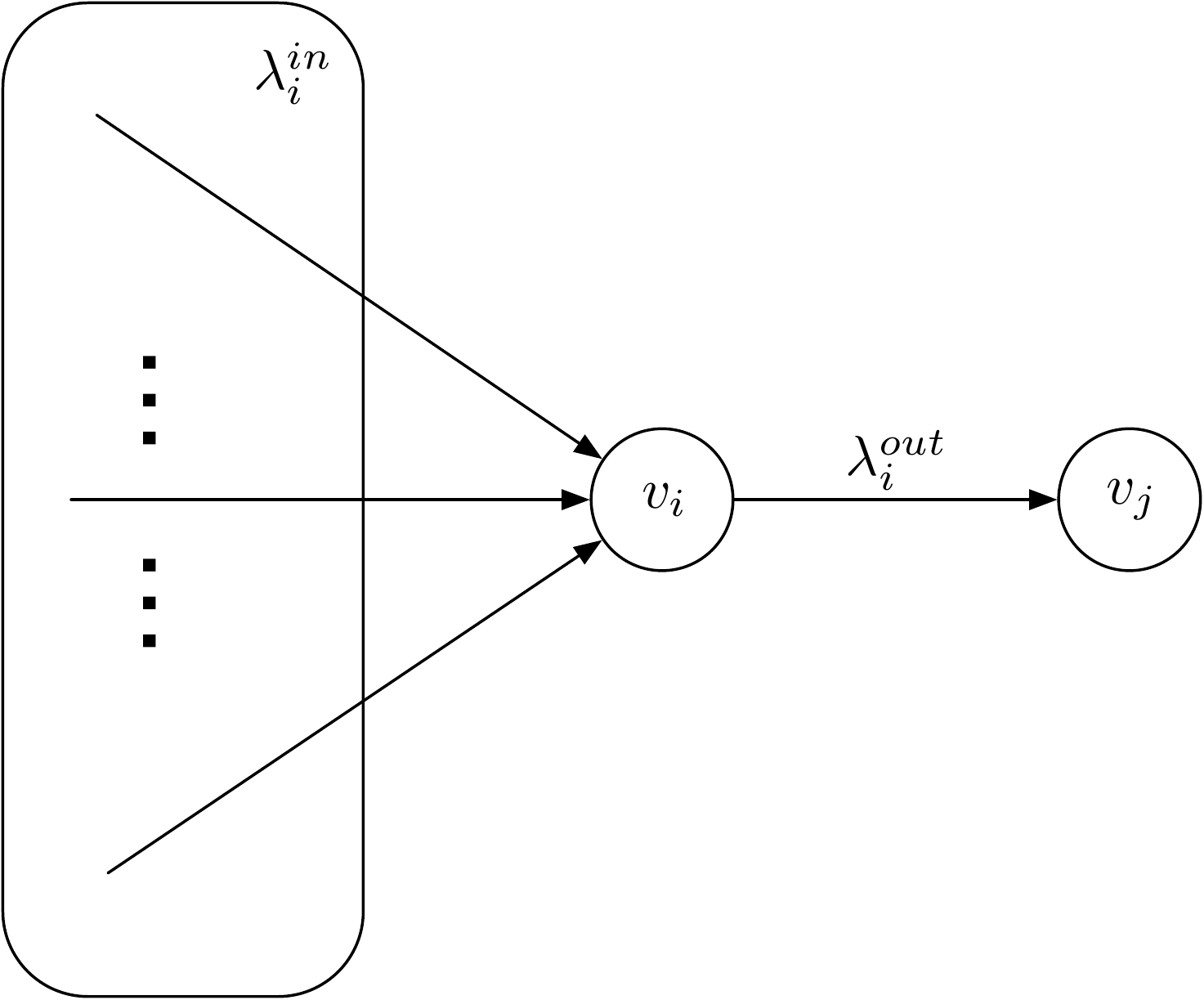}
		\caption{}
		\label{fig:proof2}
	\end{subfigure}
	\label{fig:proof}
	\caption{The potential tier of a vertex depends on its direct predecessors. The optimal tier of a vertex depends on its input and output size.}
\end{figure}

Our main objective is to minimize the DNN inference time by leveraging the 
collaborative computation provided by three computing tiers. 
To achieve this goal, we propose to split the DNN model into three parts 
that are executed over device, edge, and cloud.
Given a DNN layer, the total latency includes 
the DNN layer processing time and its input data transmission delay.
Theoretically, for $v_i \in \mathcal{V}$ 
and a set of its direct predecessors~$\mathcal{V}_i^p$, 
the total latency is~$t_i^{l_i} + \sum_{v_h \in \mathcal{V}_i^p} t_{hi}^{[l_h, l_i]}$. 
The optimal tier of a vertex~$v_i$ is decided by comparing the latencies 
of attaching~$v_i$ to different computing tiers. 
We regard the tier that is possible to yield the smallest latency as a
\textit{potential tier}. A set of potential tiers of $v_i$ is denoted by $\Gamma_i$, 
where $\Gamma_i \subset \{d, e, c\}$. 

To this end, mathematically, our major goal is to split the DAG in~\eqref{eqn:dag} into 
three sub-graphs by assigning each vertex to one of the three computing tiers 
and minimizing the total latency 
\begin{equation*}
\label{eqn:totallatency}
	\Theta(v_0, v_1, \cdots, v_n) = \sum_{v_i\in \mathcal{V}} t_i^{l_i} + \sum_{(v_i, v_j)\in \mathcal{L}} t_{ij}^{[l_i, l_j]}.
\end{equation*}
Each sub-graph contains a subset of $\mathcal V$ that have the same optimal tier.
However, partitioning a DAG according to multiple vertex weights 
and link weights falls into an NP-hard problem~\cite{hartmanis1982computers, nossack2014branch}. 
Thus, heuristics are essential to partition the DAG. 

With the DAG in~\eqref{eqn:dag} representing a DNN, 
the \textbf{h}orizontal \textbf{p}artition \textbf{a}lgorithm (HPA) first computes the 
longest distance from $v_{0}$ to $v_{i}$, denoted by $\delta(v_i)$, $\forall v_i \in \mathcal{V}$.
We get the longest distance with the dynamic programming method 
mentioned in~\cite{sedgewick2011algorithms}. 
The time complexity of the method is $\mathcal{O}(|\mathcal{V}|+|\mathcal{L}|)$. 
Subsequently, we define a partition of $\mathcal V$ by 
\begin{equation*}
	\mathcal{Z}_q := \{v_i: \delta(v_i)=q, v_i\in \mathcal{V}\}, \quad q = 0, 1, \dots, n-1,
\end{equation*}
and HPA arranges $v_i$ to the graph layer $\mathcal{Z}_q$. 
To illustrate, in Fig.~\ref{fig:inceptiondag}, 
HPA assigns the vertices to~$7$ graph layers that are 
$\mathcal{Z}_0=\{v_0\}$, 
$\mathcal{Z}_1=\{v_1\}$, $\mathcal{Z}_2=\{v_2, v_3, v_4, v_5\}$, 
$\mathcal{Z}_3=\{v_6, v_7, v_8, v_9\}$, $\mathcal{Z}_4=\{v_{10}\}$, 
$\mathcal{Z}_5=\{v_{11}, v_{12}\}$, $\mathcal{Z}_6=\{v_{13}\}$. 
In graph layer $\mathcal{Z}_q$ where $q \in \mathbb{Z}^+$, 
to decide the optimal tier of each vertex in $\mathcal{Z}_q$, 
HPA must assign all vertices in $\mathcal{Z}_{q-1}$ to their optimal tiers.
By mathematical induction, we know that HPA starts from graph layer $\mathcal{Z}_0$ 
and calculates the optimal tiers layer by layer. 

For a vertex~$v_i$, its potential tiers 
and the tiers of its direct predecessors 
have the relation described in the Proposition~\ref{prop:potential}. 
As an example, if the direct predecessors of vertex~$v_i$ are assigned to an edge node, 
then the potential tiers for $v_i$ are edge and cloud. 
\begin{proposition}
    \label{prop:potential}
	For a vertex~$v_i \in \mathcal{V}$ and a set of its direct predecessors 
	$\mathcal{V}_i^p=\{v_{h^1}, v_{h^2}, \cdots, v_{h^m}\}$, 
	the potential tier $l_i$ of $v_i$ has the subsequent relation with 
	the tiers of its direct predecessors, i.e.,
	$\max\{l_{h^1}, l_{h^2}, \cdots, l_{h^m}\}\succeq l_i$.
\end{proposition}

\begin{proof}
	We prove the proposition~\ref{prop:potential} by contradiction. 
	Consider a vertex~$v_i$ with multiple direct predecessors 
	$\{v_{h^1}, \cdots, v_{h^j}, v_{h^{j+1}}, \cdots, v_{h^m}\}$ shown in Fig.~\ref{fig:proof1}. 
	We assume that the DNN layers represented by $\{v_{h^1}, \cdots, v_{h^j}\}$ 
	are assigned to a cloud node 
	and the DNN layers denoted by $\{v_{h^{j+1}}, \cdots, v_{h^m}\}$ are assigned to an edge node. 
	Suppose that we assign the DNN layer described by $v_i$ to a device node, 
	which means \mbox{$l_i\succ\max\{l_{h^1}, l_{h^2}, \cdots, l_{h^m}\}$}. 
	The total latency of $v_i$ is $t_i^{d} + \sum_{x=1}^{j} t_{h^xi}^{[c, d]} + \sum_{y=j+1}^{m} t_{h^yi}^{[e, d]}$, 
	which is larger than the delay $t_i^{e} + \sum_{x=1}^{j} t_{h^xi}^{[c, e]}$ 
	when the DNN layer signified by $v_i$ is assigned to an edge node. 
	Assigning $v_i$ to the device tier cannot yield the smallest latency. By contradiction, 
	we prove \mbox{$\max \{l_{h^1}, l_{h^2}, \cdots, l_{h^m}\}\succeq l_i$}.
\end{proof}

\begin{table}[t]
	\centering
	\caption{The total latencies of processing $v_i$ and $v_j$.}
	\setlength{\extrarowheight}{2pt}
	\begin{tabular}{|| C{5 em} | C{5 em} | C{14 em}||}
		\hline
		\textbf{Location of $v_i$} & \textbf{Location of $v_j$} & \textbf{Total Latency} \\ 
		\hline
		device & device & $t_i^d + t_j^d$  \\ 
		\hline
		device & edge & $t_i^d + t_j^e + \lambda_i^{out}/\sigma_{de}$  \\
		\hline
		edge & edge & $t_i^e + t_j^e + \lambda_i^{in}/\sigma_{de}$ \\
		\hline
		edge & cloud & $t_i^e + t_j^c + \lambda_i^{in}/\sigma_{de} + \lambda_i^{out}/\sigma_{ec}$ \\
		\hline
		cloud & cloud & $t_i^c + t_j^c + \lambda_i^{in}/\sigma_{dc}$ \\
		\hline
		device & cloud & $t_i^d + t_j^c + \lambda_i^{out}/\sigma_{dc}$ \\
		\hline
	\end{tabular}
    \label{tab:optimaltier}
\end{table}
 
We now derive the optimal tier selection strategy for a vertex~$v_i \in \mathcal{V}$. 
Intuitively, for a single vertex $v_i \in \mathcal{V}$, the optimal tier $l_i^{opt}$ is calculated as: 
\begin{equation}
\label{eqn:optimaltier}
l_i^{opt}=\mathop{\arg\min}_{l_i \in \Gamma_i} (t_i^{l_i} + \sum_{v_h \in \mathcal{V}_i^p} t_{hi}^{[l_h, l_i]}).
\end{equation}
Specially, $l_0^{opt}=d$ for the virtual input vertex~$v_0$.
However, Equation~\eqref{eqn:optimaltier} only selects the local optimal location for $v_i$. 
To further optimize the selection, 
we take the input and output size of a DNN layer into consideration. 
Assuming that a vertex~$v_i$ has multiple direct successors, 
if we place the DNN layers represented by these successors 
on the same computation node and give them the same input, 
we call the successor representing the DNN layer with the longest processing time 
as the \textit{largest direct successor} of $v_i$.
Fig.~\ref{fig:proof2} depicts two vertices $v_i \in \mathcal{V}$ and $v_j \in \mathcal{V}$, 
where $v_j$ is the largest direct successor of $v_i$ and $\Gamma_i=\{d, e, c\}$. 
The total input size of $v_i$ is $\lambda_i^{in}$, and its output size is $\lambda_i^{out}$. 
The network bandwidth between a device node and an edge node, an edge node and a cloud node, 
a device node and a cloud node are $\sigma_{de}, \sigma_{ec}, \sigma_{dc}$ respectively. 
We list the total latencies of allocating the DNN layers indicated by $v_i$ and $v_j$ to 
computation nodes at different tiers in TABLE~\ref{tab:optimaltier}. 
The inputs of $v_i$ are from the device tier.
Heuristically, HPA selects the optimal tier of $v_i$ via the following mechanism. 
On one hand, if $\lambda_i^{in} > \lambda_i^{out}$, $v_i$'s optimal tier $l_i^{opt}$ 
is computed via Equation~\eqref{eqn:optimaltier}. 
On the other hand, if $\lambda_i^{in} \leq \lambda_i^{out}$, 
HPA chooses the largest direct successor of $v_i$
and computes the total latencies of processing $v_i$ 
and its largest direct successor as in TABLE~\ref{tab:optimaltier}. 
According to the smallest value among these total latencies, HPA selects the optimal tier of $v_i$. 

\begin{figure}[t]
	\centering
	\includegraphics[width=0.2\textwidth]{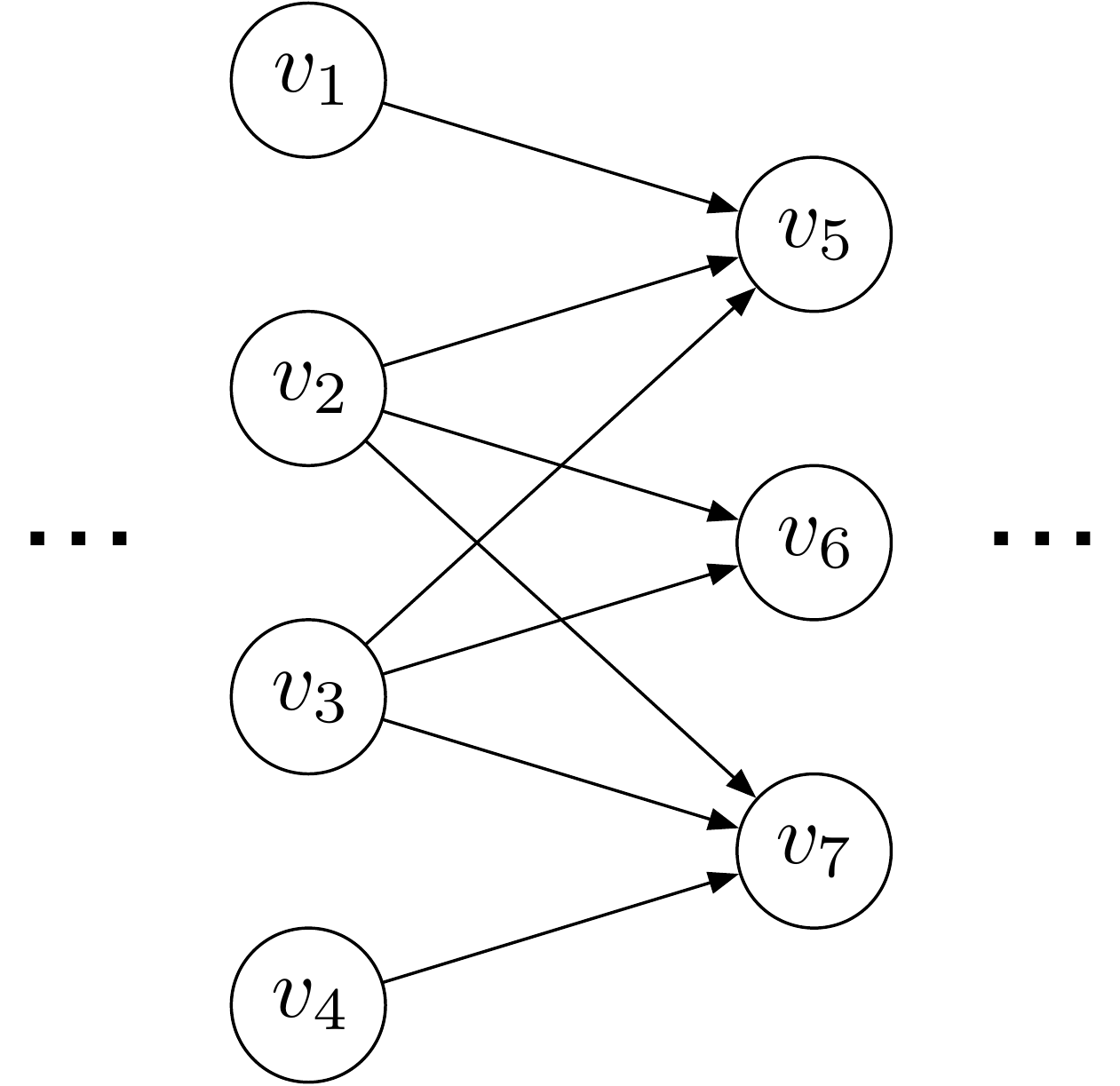}
	\caption{Illustration of SIS relation between vertices.}
	\label{fig:sisexa}
\end{figure}

Given a vertex~$v_i \in \mathcal{Z}_q$, 
we refer to a vertex~$v_j$ as the \textit{\textbf{s}ubset \textbf{i}nput \textbf{s}ibling} (SIS) 
vertex of $v_i$ if $\mathcal{V}_j^p \subset \mathcal{V}_i^p$. 
We clarify the SIS relation with an example in Fig.~\ref{fig:sisexa} 
where $7$ vertices are connected via directed links. 
In the example, $v_6$ is the SIS vertex of $v_5$ since $\mathcal{V}_6^p \subset \mathcal{V}_5^p$, 
whereas $v_7$ is not the SIS vertex of $v_5$ 
in that $\mathcal{V}_7^p \not\subset \mathcal{V}_5^p$.
After deciding the optimal tier for all vertices in layer $\mathcal{Z}_q$, 
for each $v_i \in \mathcal{Z}_q$, if the optimal tier of $v_i$'s SIS vertex~$v_j$ is before 
$l_i^{opt}$ (i.e., $l_j^{opt} \succ l_i^{opt}$), 
HPA updates the optimal tier of the SIS vertex to $l_i^{opt}$. 
We regard this approach as \textit{SIS update} that follows the Proposition~\ref{prop:sisvertexupdate}. 
%In particular, for a vertex, if the optimal tiers of its direct predecessors are all $c$ (i.e., cloud), 
%the optimal tier of the vertex is also $c$. 
%HPA stops when it determines the optimal tiers for all the vertices in~$\mathcal{G}$. 

\begin{proposition}
    \label{prop:sisvertexupdate}
    SIS update optimizes the processing delay and the transmission latency of a SIS vertex.
\end{proposition}
\begin{proof}
	Given a vertex~$v_i$ and one of its SIS vertices $v_j$ where $l_j^{opt} \succ l_i^{opt}$. 
	Since the inputs of $v_i$ are already transmitted 
    to the tier $l_i^{opt}$, therefore relocating the SIS vertex 
    that is at the previous tier $l_j^{opt}$ to $l_i^{opt}$ reduces the processing time 
    and brings no transmission delay overhead. 
\end{proof} 

\begin{algorithm}[t]
	\caption{Horizontal Partition Algorithm: HPA()}
	\label{alg:hpa}
	\KwData{DAG: $\mathcal{G=(V, L)}$ \; \hspace{0.83cm} Vertex weights: $\mathcal{T}_{\star}$ \; \hspace{0.87cm} Link weights: $\mathcal{T}_{\dagger}$ .}
	\KwResult{Optimal tiers $l_i^{opt}$ where $i=1, 2, \cdots, |\mathcal{V}|$ .}
	$\mathcal{Q} \leftarrow \mathtt{get\_longest\_path(} \mathcal{G} \mathtt{)}$ \;
	$\Delta \leftarrow \mathtt{get\_graph\_layer(} \mathcal{Q}, \mathcal{V} \mathtt{)}$ \;
	\ForEach{$\mathcal{Z}_q$ in $\Delta$}{
		\ForEach{$v_i$ in $\mathcal{Z}_q$}{
			$\Phi_i \leftarrow \mathtt{get\_pred\_loc(} v_i\mathtt{)}$ \;
			$\Gamma_i \leftarrow \mathtt{get\_loc\_choice(} \Phi_i \mathtt{)}$ \;
			\uIf{$\Gamma_i = \{c\}$}{
				$l_i^{opt} \leftarrow c$ \;
			}
			\Else{
				$l_i^{opt} \leftarrow \mathtt{get\_opt\_loc(}v_i, \Gamma_i, \mathcal{T}_{\star}, \mathcal{T}_{\dagger})$ \;
			}
		}
		$\mathtt{sis\_update(}\mathcal{Z}_q\mathtt{)}$\;
	}
\end{algorithm}

We now show the HPA in Algorithm~\ref{alg:hpa}. 
The algorithm calls~$\mathtt{get\_longest\_path()}$ to calculate 
the longest path from $v_0$ to all vertices in $\mathcal{G}$. 
With the result of the longest path, 
HPA constructs the graph layers by invoking the function $\mathtt{get\_graph\_layer()}$.
In each graph layer~$\mathcal{Z}_q$, HPA computes the optimal tiers for all vertices.
For every vertex~$v_i$ in the layer, the algorithm uses 
the function~$\mathtt{get\_pred\_loc()}$ to get the 
optimal tiers of $v_i$'s direct predecessors. 
Next, it leverages the idea presented in Proposition~\ref{prop:potential} 
and employs~$\mathtt{get\_loc\_choice()}$ to get all potential tiers of $v_i$. 
In particular, if the potential tier of $v_i$ is $c$, 
the optimal tier $l_i^{opt}$ is $c$. 
Otherwise, HPA computes the optimal tier for $v_i$ via function $\mathtt{get\_opt\_loc()}$ 
which leverages the heuristics in the optimal tier selection strategy. 
After the calculation of all vertices in the graph layer $\mathcal{Z}_q$,
$\mathtt{sis\_update()}$ performs SIS update for all vertices in $\mathcal{Z}_q$. 
HPA stops when it finishes processing all the graph layers. 

Resource changes and network dynamics lead to 
variations of DNN layer processing time and input data transmission latencies, 
which further affect the optimal locations to process DNN layers. 
Assuming that the optimal tier of a vertex changes,
HPA can accommodate the modification by locally adjusting the optimal tiers of its SIS vertices, 
its direct successors, and the SIS vertices of its direct successors. 
For instance, in Fig.~\ref{fig:inceptiondag}, supposing that $l_6^{opt}$ of $v_6$ 
changes to a different value, HPA recalculates $l_{10}^{opt}$ of $v_{10}$ 
since $v_{10}$ is the direct successor of $v_6$. 
To avoid constantly calculating the optimal tier to respond to the resource and network fluctuation, 
we can set upper and lower thresholds to limit the scope of the optimal tier alteration. 
HPA only recalculates the optimal tiers 
when DNN layer processing time or the network bandwidth is outside of the threshold range.

\subsection{Vertical Separation Module} 

The preliminary experimental results in TABLE~\ref{tab:edgebottleneck} show 
the inference latencies of DNNs after HPA given an input image of $3 \times 224 \times 224$. 
The device node is an NVIDIA Jetson Nano 2GB Developer Kit~\cite{nvidiajetsonnano2gb}, 
the edge node is a Linux machine with Intel Core i7-8700 CPU and 8 GB system memory, 
and the cloud node is a Linux server with NVIDIA GeForce RTX 2080 Ti GPU and 256 GB system memory. 
We observe that the processing time of the edge node is longer than that of the cloud node, 
causing the cloud node to be idle and waiting for the results from the edge node in the inference pipeline. 
The edge node becomes the bottleneck of the synergistic inference. 
To accelerate the inference at the edge tier, 
we design a parallel processing strategy to avoid the bottleneck condition. 
The parallelization option is not suitable for the resource stringent device node 
since processing raw input in parallel incurs privacy concerns and amount of communication overhead. 
Next, we introduce our parallel processing method. 

\begin{table}[t]
	\centering
	\caption{The synergistic inference time at three nodes.}
	\setlength{\extrarowheight}{2pt}
	\begin{tabular}{|| C{6 em} |  C{6 em} | C{6 em} | C{6 em}||}
		\hline
		\textbf{DNNs} & \textbf{Device Node (millisecond)} & \textbf{Edge Node (millisecond)} & \textbf{Cloud Node (millisecond)} \\ 
		\hline
		AlexNet & 2.2 & 3.6 & 1.4 \\ 
		\hline
		VGG-16 & 5.7 & 46.7 & 0.5 \\
		\hline
		ResNet-18 & 6.1 & 7.5 & 0.5 \\
		\hline
		Darknet-53 & 27.9 & 48.1 & 0.1\\
		\hline
		Inception-v4 & 21.4 & 46.4 & 16.7 \\
		\hline
	\end{tabular}
    \label{tab:edgebottleneck}
\end{table} 

The parameters of a convolutional layer are a set of learnable filters~\cite{goodfellow2016deep}.
Each filter is a weighted tensor spatially defined by its hyper-parameters 
that are width, height, and depth. 
An input feature map is the input activation for a given filter. 
The number of input feature maps is the same as the filter depth in a convolutional layer. 
Besides, a convolutional layer contains two hyper-parameters 
that are filter stride and padding for input feature maps. 
If padding exists, the convolutional layer adds entries to the borders of an input feature map, 
resulting in an input feature map with paddings. 
The convolutional layer systematically performs a dot product 
between the entries of the padded input feature maps and the filter. 
The results are assembled to an output feature map. 
We call this process a convolution operation. 
In order to optimize the inference latency at the resource constrained edge node, 
we leverage the idea of separating a sequence of 
correlated input feature maps to multiple feature map stacks. 
This idea is firstly proposed in DeepThings~\cite{zhou2019adaptive}.
However, DeepThings does not consider input feature maps with paddings, 
leading to the precision loss that affects the inference accuracy. 
To settle this issue, we derive a parallel convolutional layer inference module 
without loss of accuracy referred to as \textbf{v}ertical \textbf{s}eparation \textbf{m}odule (VSM). 

Given a sequence of $k$ convolutional layers, 
we describe each convolutional layer as $c_i$ where $i=1, 2, \cdots, k$. 
The input feature maps of layer $c_i$ have the dimension of 
$\mathcal{W}_i \times \mathcal{H}_i \times \mathcal{D}_i$ (width $\times$ height $\times$ depth). 
For the filter of $c_i$, we designate the filter size 
as $\mathcal{F}_i^w \times \mathcal{F}_i^h \times \mathcal{D}_i$ 
(width $\times$ height $\times$ depth) 
with a horizontal stride of $\mathcal{S}_i^w$ 
and a vertical stride of $\mathcal{S}_i^h$. 
The padding is $\mathcal{P}_i^w$ horizontally and $\mathcal{P}_i^h$ vertically.
Consequently, the input feature map size of layer $c_{i}$ 
has the coming relation with the input feature map size of layer $c_{i-1}$: 
\begin{equation} \label{eqn:fmaps}
\begin{split}
\mathcal{W}_{i} & = \dfrac{\mathcal{W}_{i-1} - \mathcal{F}_{i-1}^w + 2 \times \mathcal{P}_{i-1}^w}{\mathcal{S}_{i-1}^w} + 1, \\
\mathcal{H}_{i} & = \dfrac{\mathcal{H}_{i-1} - \mathcal{F}_{i-1}^h + 2 \times \mathcal{P}_{i-1}^h}{\mathcal{S}_{i-1}^h} + 1.
\end{split}
\end{equation} 

We divide the input feature maps of layer $c_i$ into $\mathcal{A} \times \mathcal{B} $ 
non-overlapping continuous \textit{tiles} whose depth is $\mathcal{D}_i$. 
Generally, we index an entry of the input feature map with two-dimensional coordinates. 
Therefore, we use the coordinates of the top left corner 
and the bottom right corner of a tile to locate the tile in the input feature maps. 
Mathematically, the tile $\tau_i^{(a, b)}=(\alpha_i^{(a, b)}, \beta_i^{(a, b)})$ 
at layer $c_i$ where $a=0, 1, \cdots, \mathcal{A}-1;$ \mbox{$b=0, 1, \cdots, \mathcal{B}-1$}. 
Specifically, the tile at the top left corner of 
the layer $c_i$'s input feature map is $\tau_i^{(0, 0)}$.
We represent the top left coordinate of $\tau_i^{(a, b)}$ as 
$\alpha_i^{(a, b)}=(x_{i}^{[\alpha, (a, b)]}, y_{i}^{[\alpha, (a, b)]})$ 
and the bottom right coordinate as $\beta_i^{(a, b)}=(x_{i}^{[\beta, (a, b)]}, y_{i}^{[\beta, (a, b)]})$. 
Particularly, $\alpha_i^{(0, 0)}=(0, 0)$.
In the light of Equation~\eqref{eqn:fmaps}, 
given the coordinates of $\tau_i^{(a, b)}$ $(i > 1)$, 
we compute the coordinates of the correlated tile with paddings at layer $c_{i-1}$ 
denoted by $\hat{\tau}_{i-1}^{(a, b)}=(\hat{\alpha}_{i-1}^{(a, b)}, \hat{\beta}_{i-1}^{(a, b)})$ 
where $\hat{\alpha}_{i-1}^{(a, b)}=(\hat{x}_{i-1}^{[\hat{\alpha}, (a,b)]}, \hat{y}_{i-1}^{[\hat{\alpha}, (a,b)]})$ 
and $\hat{\beta}_{i-1}^{(a, b)}=(\hat{x}_{i-1}^{[\hat{\beta}, (a,b)]}, \hat{y}_{i-1}^{[\hat{\beta}, (a,b)]})$ in following manner:
\begin{equation} \label{eqn:fmapscoord}
\begin{split}
\hat{x}_{i-1}^{[\hat{\alpha}, (a, b)]} & = \mathcal{S}_{i-1}^w \times x_{i}^{[\alpha, (a, b)]},\\
\hat{y}_{i-1}^{[\hat{\alpha}, (a, b)]} & = \mathcal{S}_{i-1}^h \times y_{i}^{[\alpha, (a, b)]},\\
\hat{x}_{i-1}^{[\hat{\beta}, (a, b)]} & = \mathcal{S}_{i-1}^w \times (x_{i}^{[\beta, (a, b)]} - 1) + \mathcal{F}_{i-1}^w,\\
\hat{y}_{i-1}^{[\hat{\beta}, (a, b)]} & = \mathcal{S}_{i-1}^h \times (y_{i}^{[\beta, (a, b)]} - 1) + \mathcal{F}_{i-1}^h.
\end{split}
\end{equation}

\begin{figure}[t]
	\centering
	\includegraphics[width=1\columnwidth]{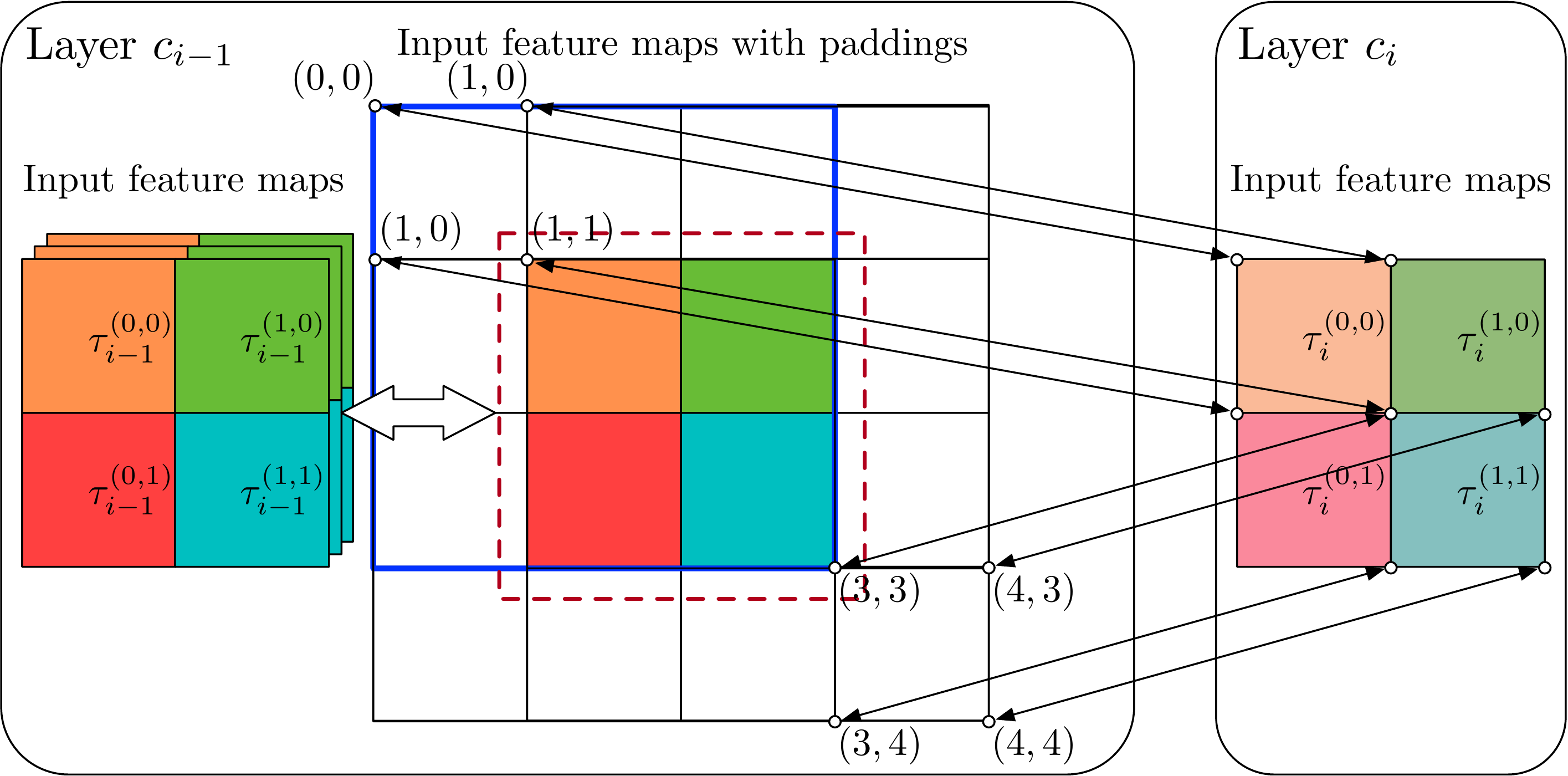}
	\caption{Layer $c_{i-1}$ has an input feature map of $2 \times 2 \times 3$ 
with a padding of $\mathcal{P}_{i-1}^w=\mathcal{P}_{i-1}^h=1$ and a stride of $\mathcal{S}_{i-1}^w$\mbox{$=\mathcal{S}_{i-1}^h=1$}. }
	\label{fig:coordalter}
\end{figure} 

Considering the paddings added to the input feature maps of layer $c_{i-1}$, 
given the coordinates of a padded tile $\hat{\tau}_{i-1}^{(a, b)}$ \mbox{$(i>1)$}, 
to compute the coordinates of tile $\tau_{i-1}^{(a, b)}$, 
we need to remove the paddings from $\hat{\tau}_{i-1}^{(a, b)}$. 
This alters the coordinates of the padded tile, 
transforming $\hat{\tau}_{i-1}^{(a, b)}$ 
to $\tau_{i-1}^{(a, b)}$.
The coordination of $\tau_{i-1}^{(a, b)}$ is described as follows: 
\begin{equation} \label{eqn:paddingalter}
\begin{split}
x_{i-1}^{[\alpha, (a,b)]} &= \max(0, \hat{x}_{i-1}^{[\hat{\alpha}, (a,b)]} - \mathcal{P}_{i-1}^w),\\
y_{i-1}^{[\alpha, (a,b)]} &= \max(0, \hat{y}_{i-1}^{[\hat{\alpha}, (a,b)]} - \mathcal{P}_{i-1}^h),\\
x_{i-1}^{[\beta, (a,b)]} &=
\begin{cases}
    \mathcal{W}_{i-1}, \quad \text{if } \hat{x}_{i-1}^{[\hat{\beta}, (a,b)]}= \mathcal{W}_{i-1} + 2\times\mathcal{P}_{i-1}^w, \\
    \max(0, \hat{x}_{i-1}^{[\hat{\beta}, (a,b)]} - \mathcal{P}_{i-1}^w), \quad \text{otherwise},
\end{cases}\\
y_{i-1}^{[\beta, (a,b)]} &=
\begin{cases}
    \mathcal{H}_{i-1}, \quad  \text{if } \hat{y}_{i-1}^{[\hat{\beta}, (a,b)]}= \mathcal{H}_{i-1} + 2\times\mathcal{P}_{i-1}^h, \\
    \max(0, \hat{y}_{i-1}^{[\hat{\beta}, (a,b)]} - \mathcal{P}_{i-1}^h), \quad \text{otherwise}.
\end{cases}\\
\end{split}
\end{equation}

Fig.~\ref{fig:coordalter} depicts the process of generating layer $c_{i-1}$'s 
input feature maps from the input feature maps of layer $c_{i}$. 
When the stride is $1\times 1$ at layer $c_{i-1}$, 
the $3 \times 3 \times 3$ filter moves one entry at a time, 
producing the input feature map of layer $c_{i}$. 
Reversely, given the input feature maps of layer $c_i$, we separate the feature maps 
into $2\times 2$ tiles. With the coordinates of the four tiles at layer $c_i$, 
we find the corresponding padded tiles
at layer $c_{i-1}$ with Equation~\eqref{eqn:fmapscoord}. 
Then, we obtain the tiles at layer $c_{i-1}$ 
by offsetting the paddings from the padded tiles via Equation~\eqref{eqn:paddingalter}. 
We regard the procedure of computing tile $\tau_{i-1}^{(a, b)}$ 
from $\tau_{i}^{(a, b)}$ as \textbf{r}everse \textbf{t}ile \textbf{c}alculation (RTC). 

\begin{algorithm}[t]
	\caption{Vertical Separation Module: VSM()}
	\label{alg:vsm}
	\KwData{$k$ correlated convolutional layers: $c_i$ where $i=1, 2, \cdots, k$\; \hspace{0.93cm}Decision of separation: $\mathcal{A}\times\mathcal{B}$ tiles\; \hspace{0.8cm} The tiles at layer $c_{k+1}$: $T_{k+1}=\{\tau_{k+1}^{(a, b)}\}$ \\\hspace{0.825cm} $(a=0, 1, \cdots,\mathcal{A}-1; b=0, 1, \cdots, \mathcal{B}-1)$ .}
	\KwResult{Coordinates of the tiles at $c_1$ .}
	\ForEach{$\tau_{k+1}^{(a, b)}$ in $T_{k+1}$}{
	    \ForEach{$i \leftarrow k$ to $1$}{
	        $\tau_i^{(a, b)} \leftarrow \mathtt{RTC(}c_i, \tau_{i+1}^{(a, b)}\mathtt{)}$ \;
	    }
	}
\end{algorithm}

To accelerate the inference at the edge tier without loss of accuracy, 
we develop VSM whose overall procedure is shown in Algorithm~\ref{alg:vsm}.
Assuming that a sequence of $k$ convolutional layers are assigned to an edge node, 
to help to describe our algorithm, we add a virtual convolutional layer $c_{k+1}$, 
whose input feature maps are the output feature maps of layer $c_k$. 
% Layer $c_{k+1}$ does not change the values or volumes of its input feature maps.  
VSM separates the input feature maps of layer $c_{k+1}$ into $\mathcal{A}\times\mathcal{B}$ 
non-overlapping continuous tiles. For each tile at layer $c_{k+1}$, 
VSM finds the coordinates of corresponding tiles from layer $c_k$ to layer $c_1$ 
via a series of RTC operations. 
We refer to a stack of correlated tiles from layer $c_1$ to $c_k$ as \textit{fused tiles}. 
Next, VSM locates the tiles at layer $c_1$ and splits the input feature maps of $c_1$ 
according to the coordinates of the tiles. 
VSM then feeds the tiles at layer $c_1$ to $\mathcal{A}\times\mathcal{B}$ edge nodes,  
each of which holds the inference parameters and hyper-parameters of the $k$ convolutional layers. 
Each edge node possesses one tile and processes a fused tile stack 
independently via convolution operations, 
generating the tile of the output feature maps at layer $c_k$. 
Finally, an edge node gathers and combines 
the $\mathcal{A}\times\mathcal{B}$ tiles of the output feature maps at layer $c_k$. 
The results are the output feature maps of layer $c_k$, which are transferred to a cloud node.
We neglect batch normalization layers and activation layers 
in the middle of two convolutional layers, since their operations do not change 
the volume of input feature maps. 
Moreover, pooling layers that are used to sub-sample feature maps 
between two convolutional layers are separated 
and fused by VSM in the same way as the convolutional layers.
Fig.~\ref{fig:vsm} shows four fused tile stacks 
of three consecutive convolutional layers. 
We assign the four fused tile stacks to four edge nodes. 
\begin{figure}[t]
	\centering
	\includegraphics[width=1\columnwidth]{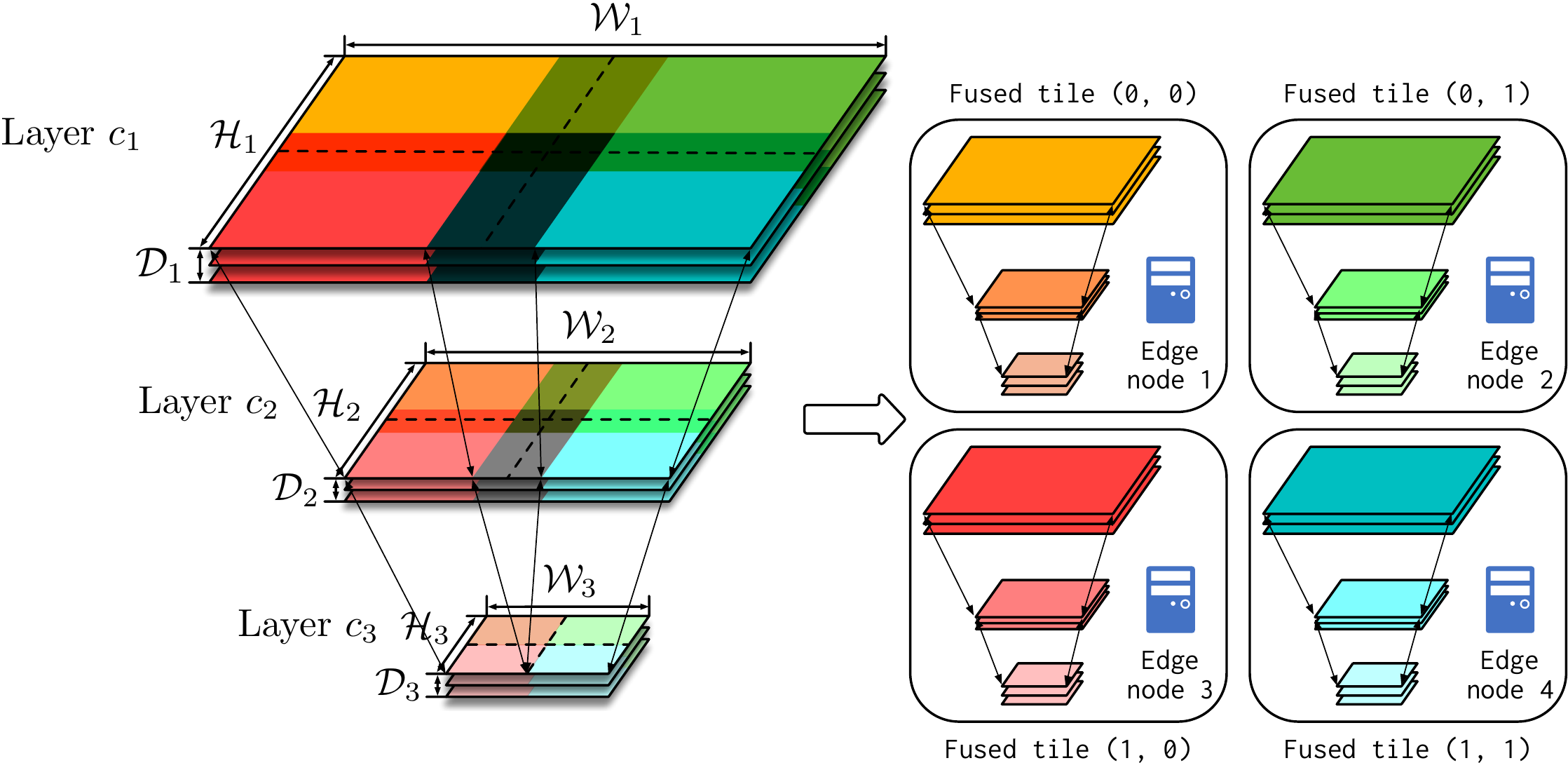}
	\caption{A sequence of input feature maps of three consecutive convolutional layers are divided into $2 \times 2$ fused tile stacks. Each fused tile stack is assigned to an edge node.}
	\label{fig:vsm}
\end{figure}

%% file: implementation.tex
\section{implementation}
\label{sec:implementation}

We introduce our implementation details in this section. 
We consider a scenario where the device node is mobile.
The edge node and the device node are in the same local area network (LAN), 
both of which connect to the cloud node via the Internet~\cite{abbas2017mobile}.
Hence, we use a Raspberry Pi 4 model B with 4 GB system memory as the device node. 
The edge node is a Linux machine with Intel Core i7-8700 CPU and 8 GB system memory. 
We employ a remote server with 
NVIDIA GeForce RTX 2080 Ti GPU and 256 GB system memory. 

\textbf{Responsibility.} We implement the profiler and the offline partition framework 
on a dedicated computation node, 
which monitors running conditions, performs dynamic partitioning according to the conditions, 
and distributes partitions to the online execution nodes. 
For the online execution procedure, 
the device node is responsible for collecting the input, 
processing through the DNN layers allocated to it, 
and transferring the output to an edge node or a cloud node. 
The edge node handles the output from the device node 
and passes the intermediate inference results to the cloud node for further computation. 
Depending on the partition decision, the inference pipeline may involve only part of the computing tiers. 

\textbf{Software stack.} We implement a client-server interface using gRPC~\cite{grpc} 
in each node to accommodate inter-process communication.  
The dedicated computation node loads a trained DNN in the ONNX~\cite{onnx} format. 
We process the DNN models with PyTorch~\cite{pytorch} 
and build the DAG representation of the DNN with NetworkX~\cite{networkx}. 
After applying HPA, we rebuild the computation graph from the partitioned DAG, 
transform the computation graph into a partial DNN 
by reconstructing the $\mathtt{forward()}$ function in PyTorch, 
and store the partial DNN in the ONNX format. 

\textbf{Datasets and models.} We evaluate D$^3$ on the ImageNet (ILSVRC2012)~\cite{deng2009imagenet} dataset
that contains 150000 validation and test images of 1000 categories. 
We compress the size of the input images to $3\times 224 \times 224$.
Then, we feed the images to the device node at 30 FPS for 100 seconds 
and test the per-image average end-to-end latency. 
We evaluate our methods across 5 widely used DNNs: 
AlexNet, VGG-16, ResNet-18, \mbox{Darknet-53}, and Inception-v4, 
all of which are trained before deployment.

%% file: eval.tex
\section{evaluation}
\label{sec:evaluation}

\begin{figure*}[t]
	\centering
	\begin{subfigure}[t]{0.24\textwidth}
		\centering
		\includegraphics[width=\textwidth]{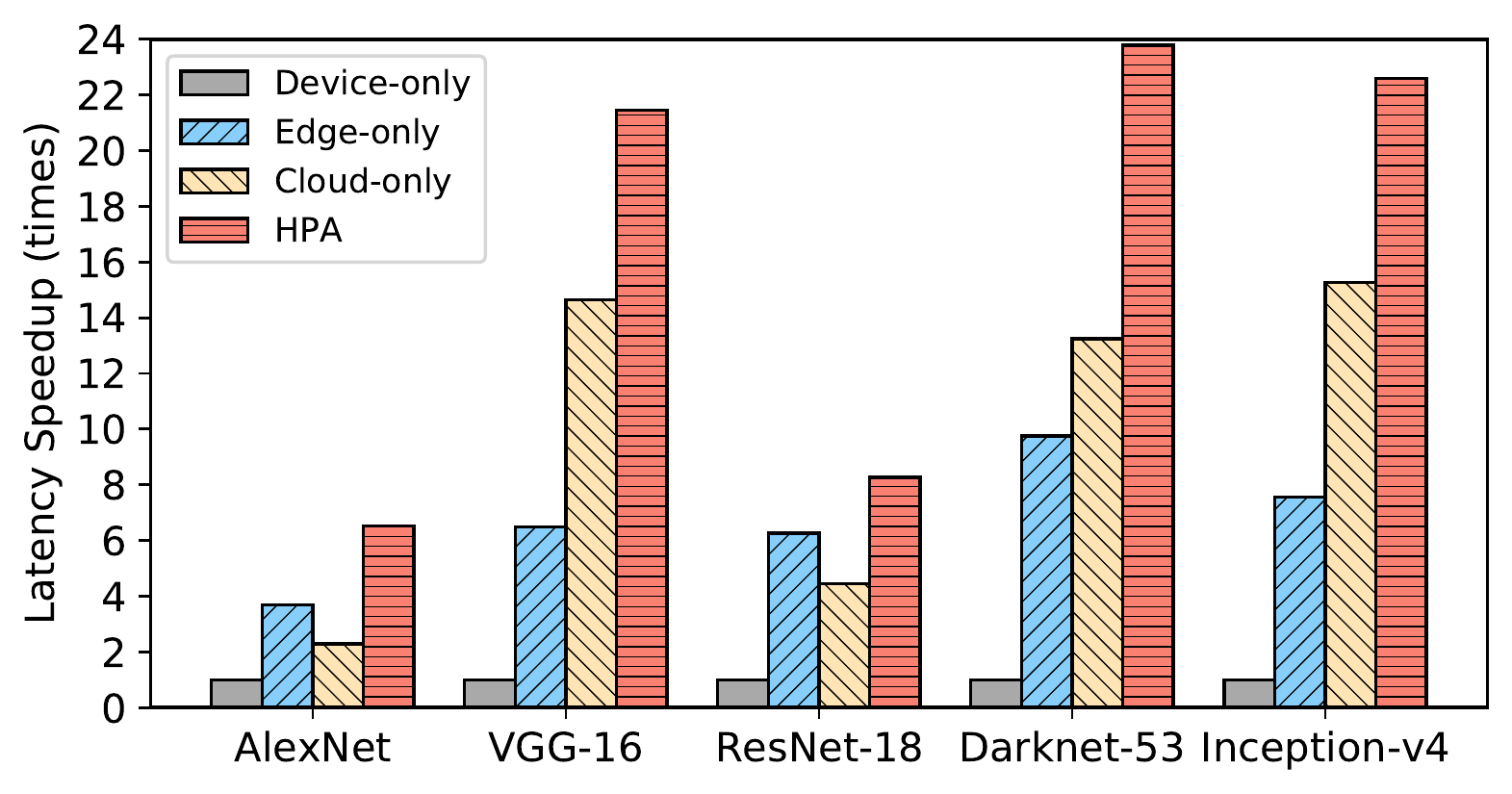}
		\caption{Wi-Fi}
		\label{plot:localspeedupwifi}
	\end{subfigure}
	\hfill
	\begin{subfigure}[t]{0.24\textwidth}
		\centering
		\includegraphics[width=\textwidth]{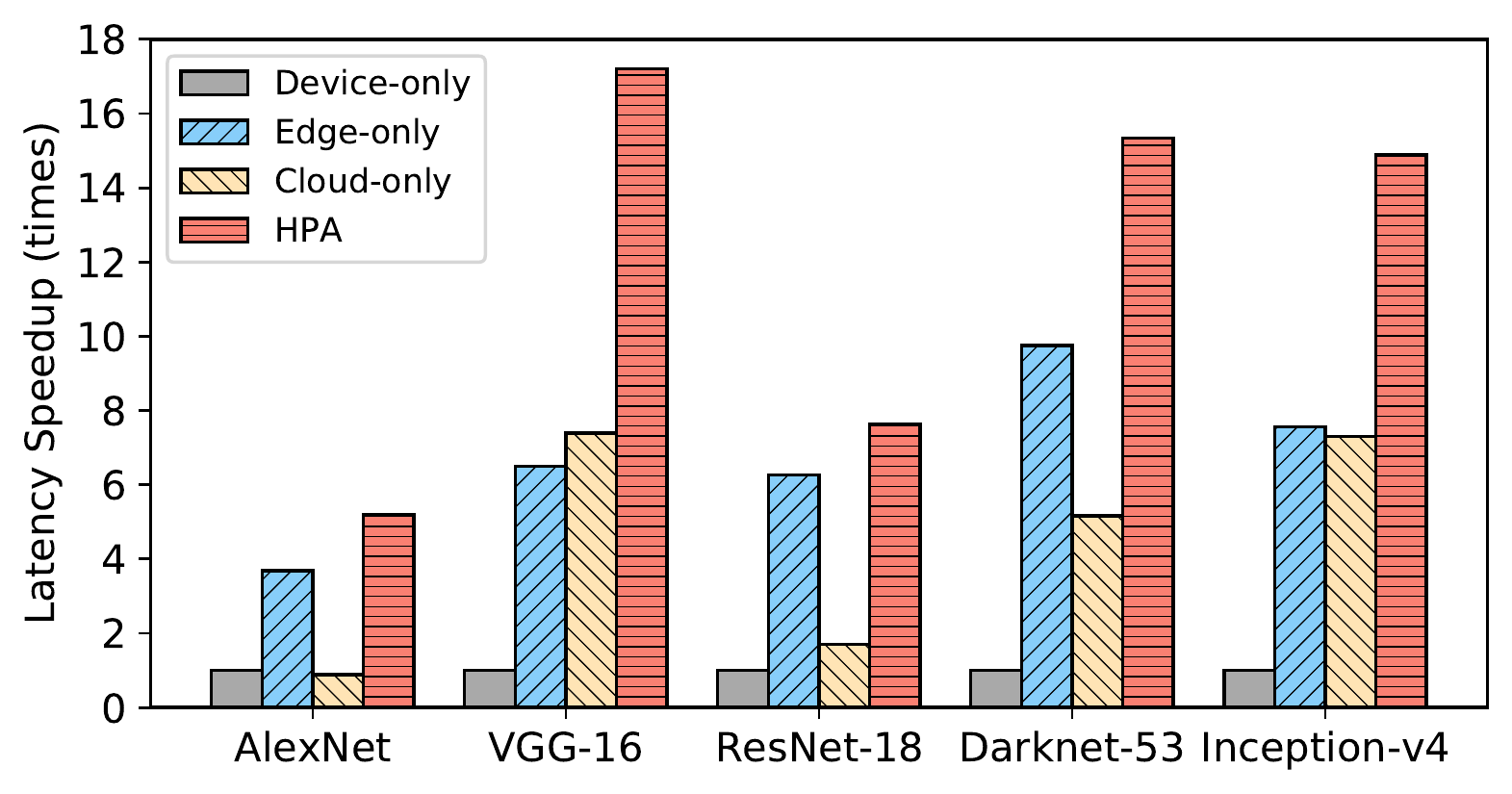}
		\caption{4G}
		\label{plot:localspeedup4g}
	\end{subfigure}
	\hfill
	\begin{subfigure}[t]{0.24\textwidth}
		\centering
		\includegraphics[width=\textwidth]{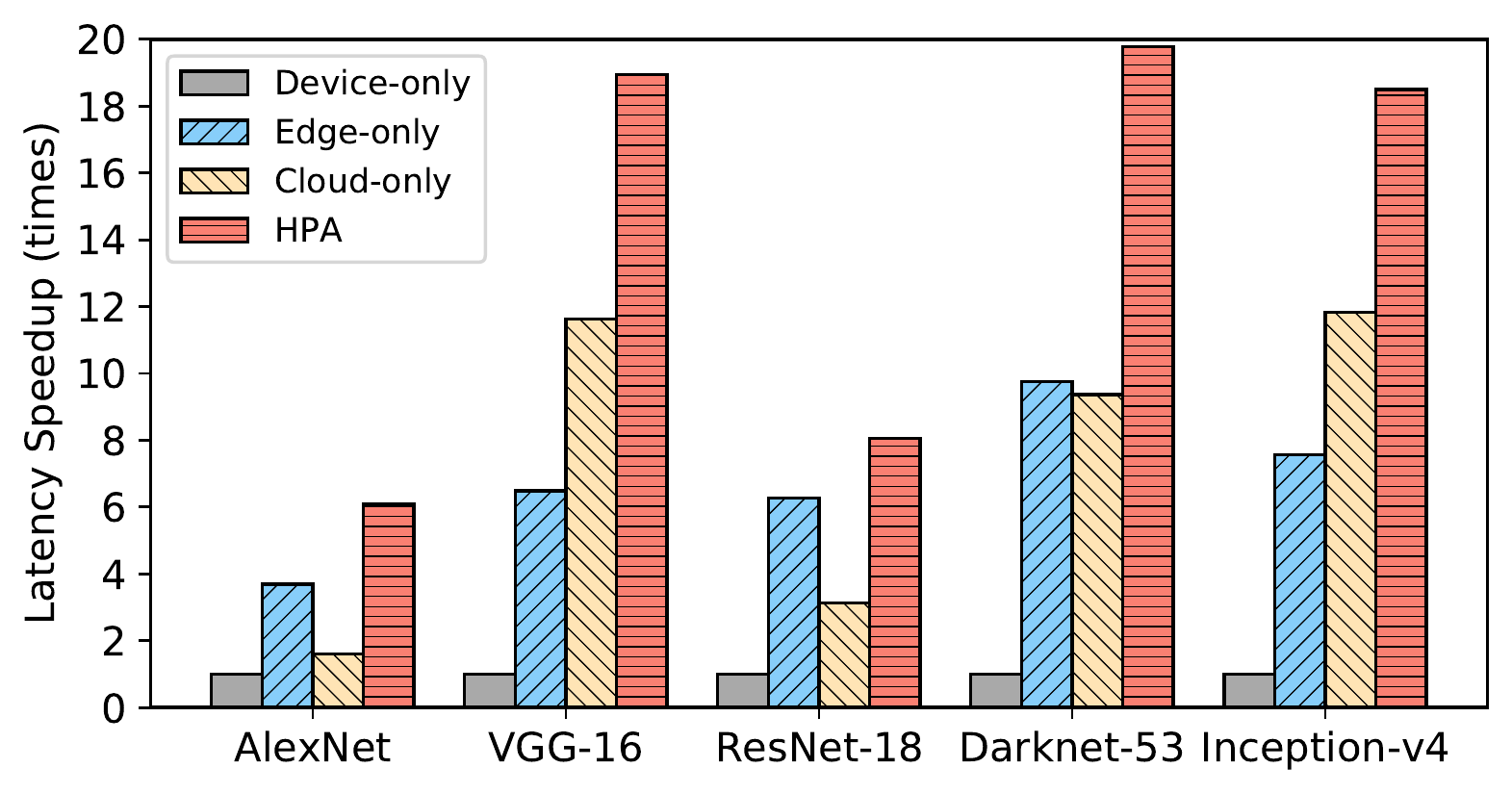}
		\caption{5G}
		\label{plot:localspeedup5g}
	\end{subfigure}
	\hfill
	\begin{subfigure}[t]{0.24\textwidth}
		\centering
		\includegraphics[width=\textwidth]{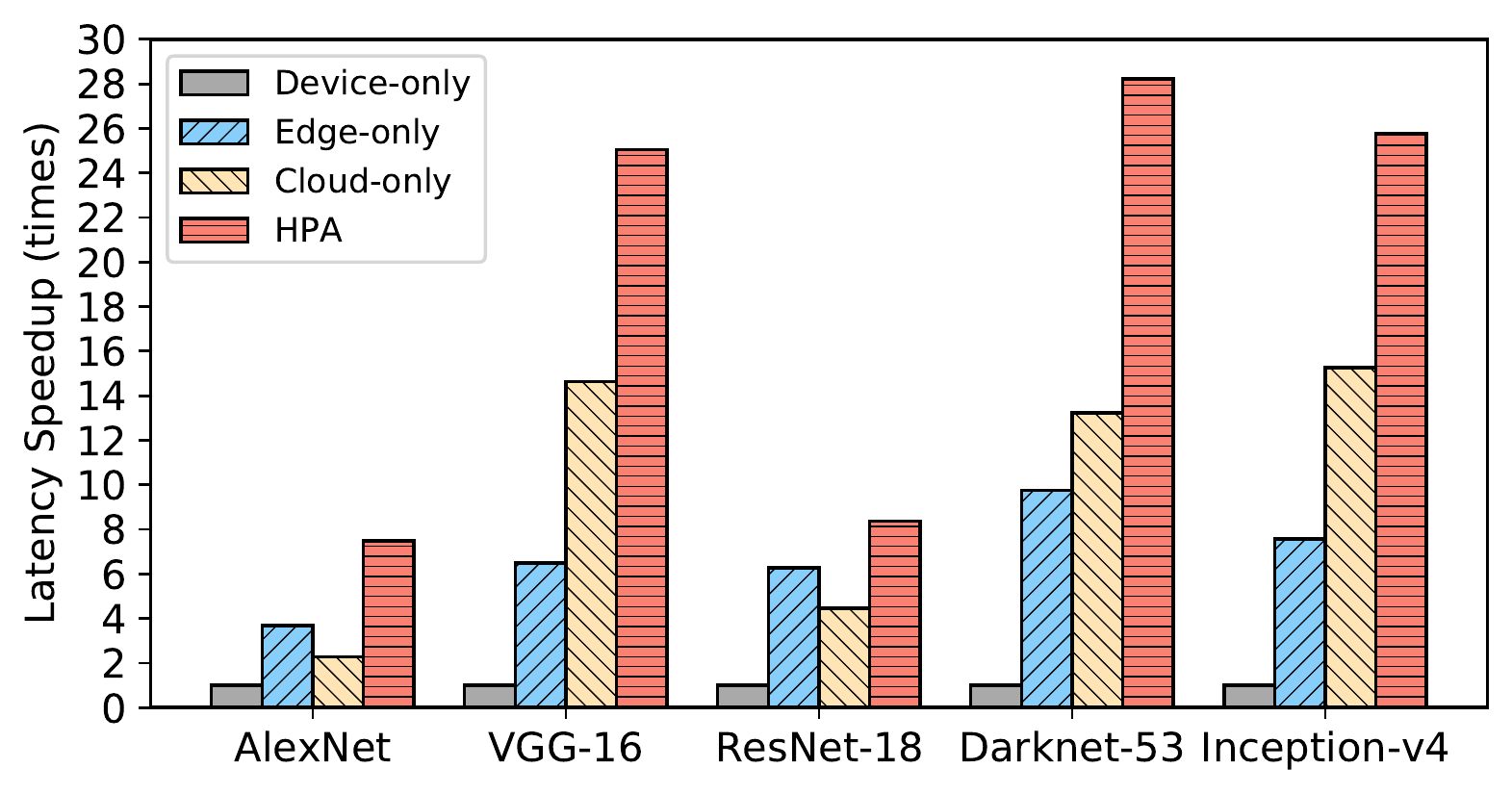}
		\caption{Optical Network}
		\label{plot:localspeedupoptical}
	\end{subfigure}
	\caption{End-to-end latency speedup comparison among HPA, device-only, edge-only, and cloud-only under different network conditions.}
	\label{plot:localspeedup}
\end{figure*}

\begin{figure*}[t]
	\centering
	\begin{subfigure}[t]{0.24\textwidth}
		\centering
		\includegraphics[width=\textwidth]{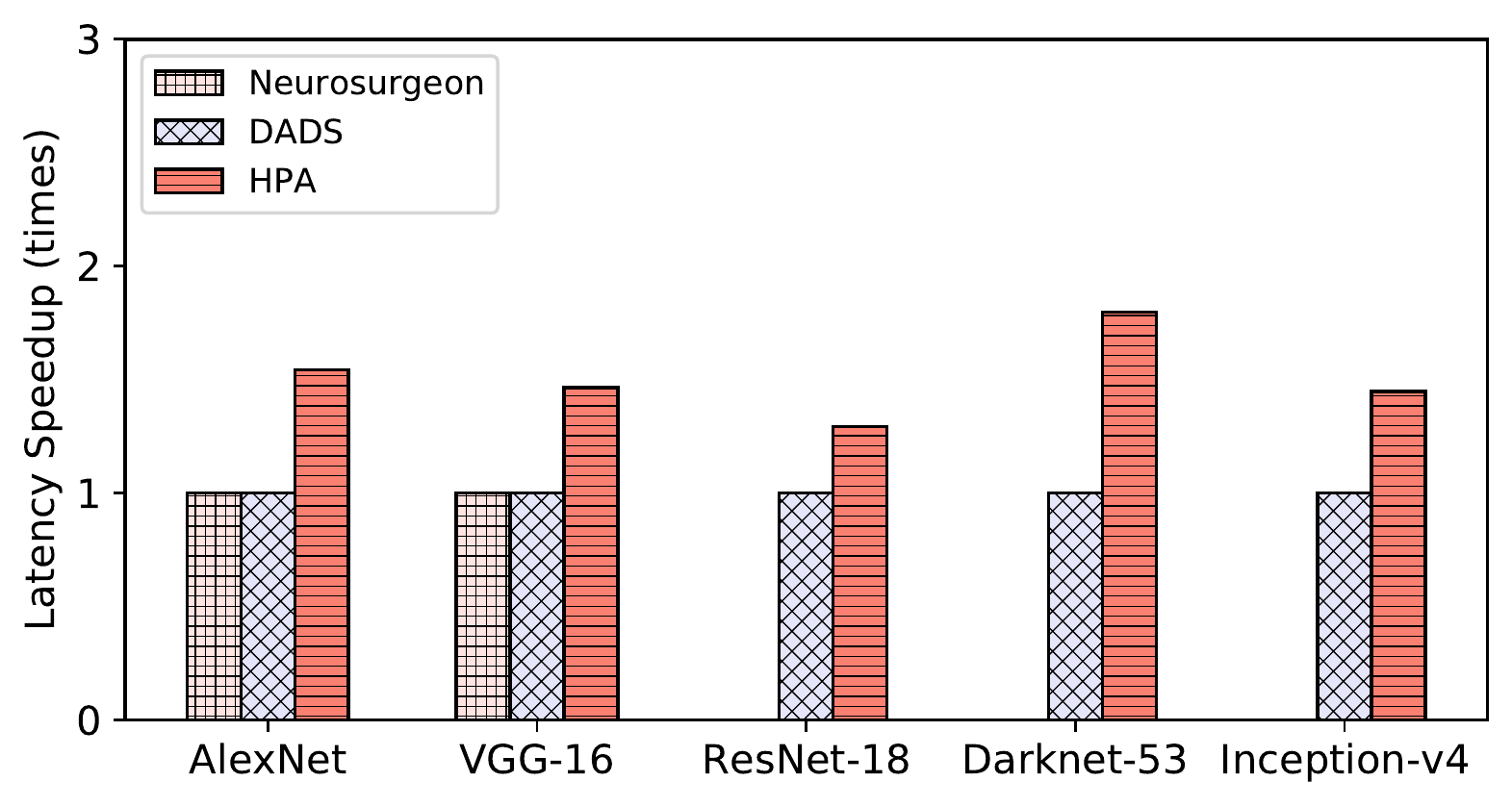}
		\caption{Wi-Fi}
		\label{plot:localspeedupdadswifi}
	\end{subfigure}
	\hfill
	\begin{subfigure}[t]{0.24\textwidth}
		\centering
		\includegraphics[width=\textwidth]{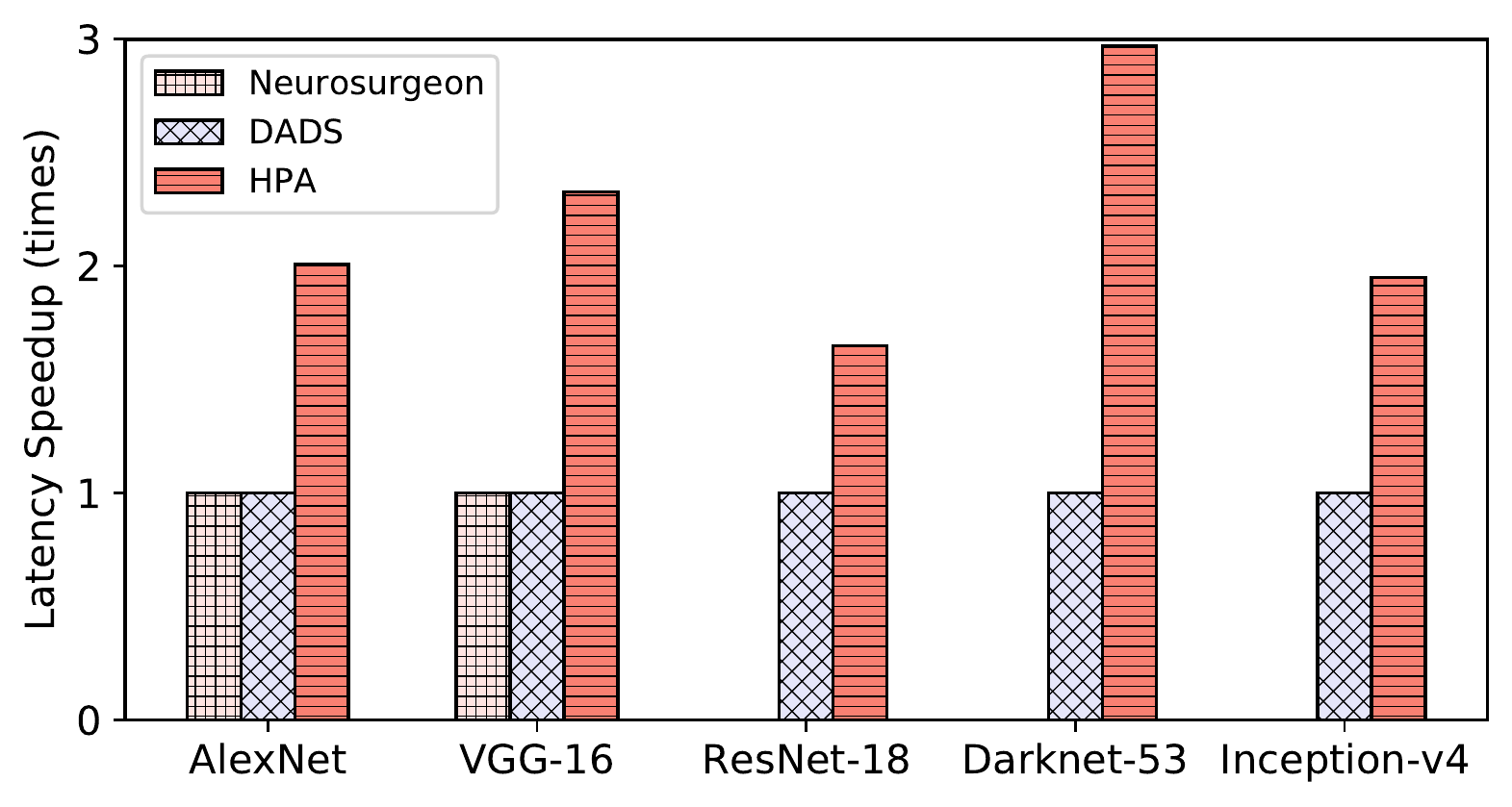}
		\caption{4G}
		\label{plot:localspeedupdads4g}
	\end{subfigure}
	\hfill
	\begin{subfigure}[t]{0.24\textwidth}
		\centering
		\includegraphics[width=\textwidth]{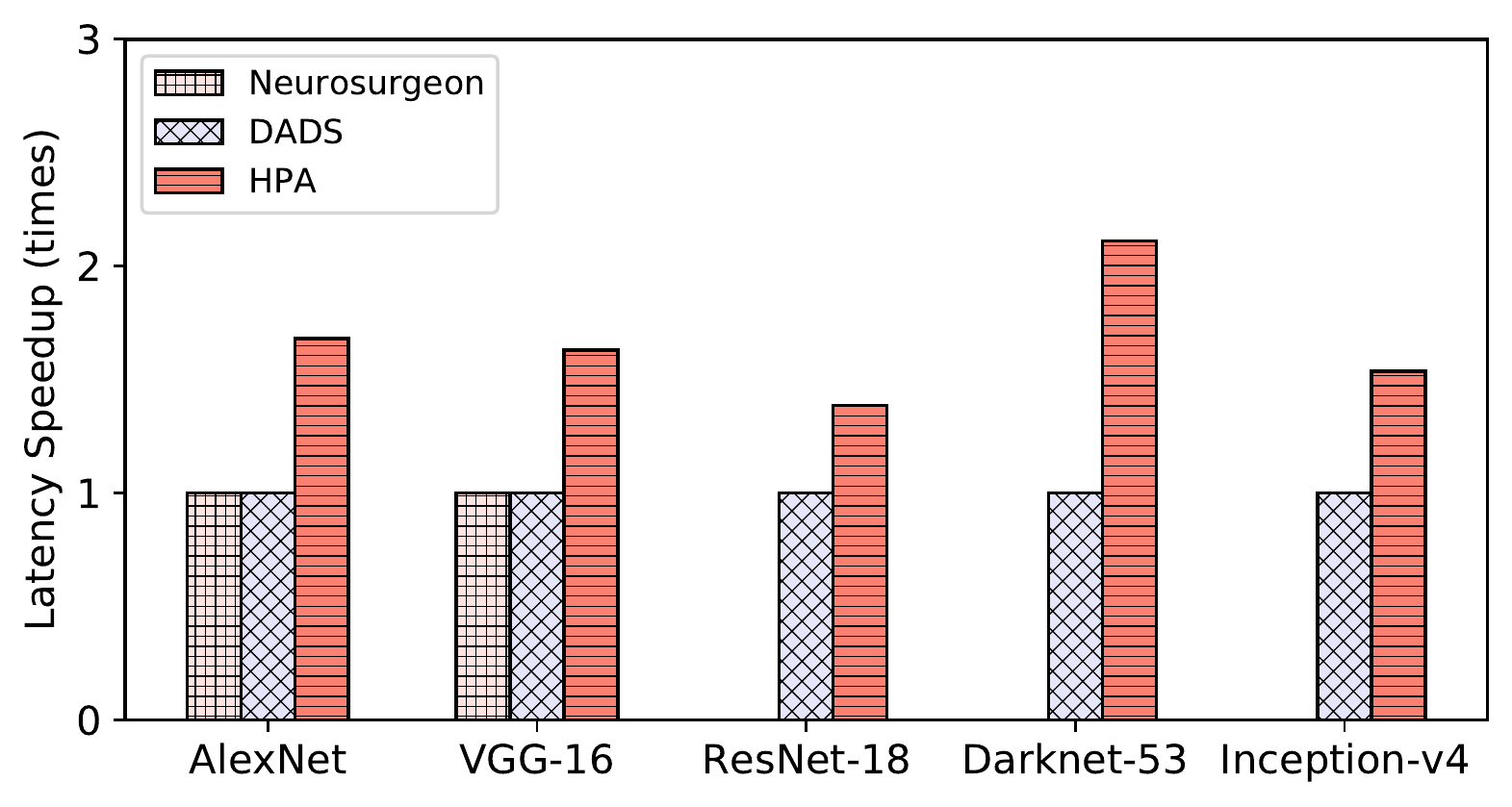}
		\caption{5G}
		\label{plot:localspeedupdads5g}
	\end{subfigure}
	\hfill
	\begin{subfigure}[t]{0.24\textwidth}
		\centering
		\includegraphics[width=\textwidth]{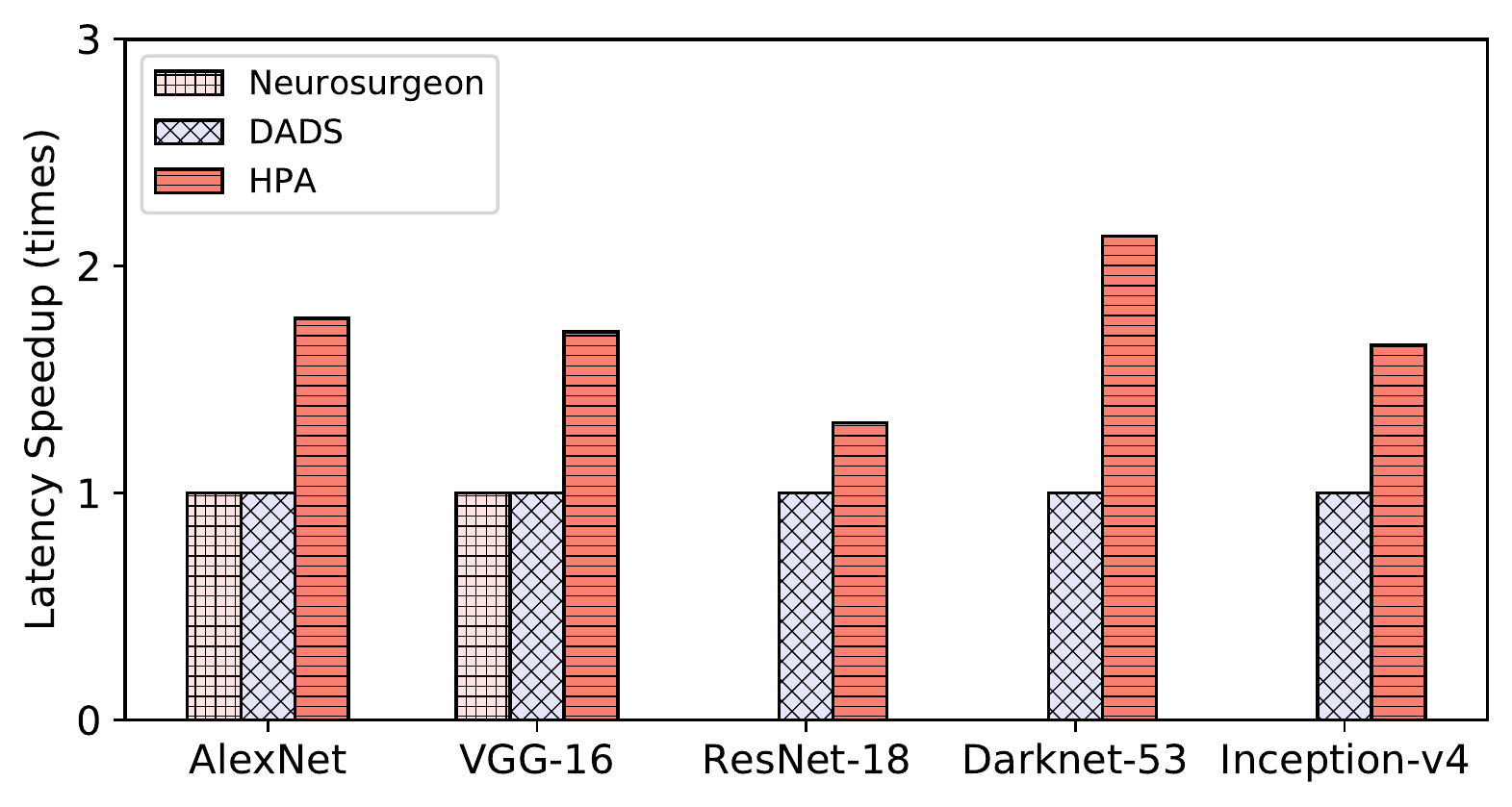}
		\caption{Optical Network}
		\label{plot:localspeedupdadsoptical}
	\end{subfigure}
	\caption{End-to-end latency speedup comparison among HPA, Neurosurgeon~\cite{kang2017neurosurgeon}, and DADS~\cite{hu2019dynamic} under different network conditions.}
	\label{plot:localspeedupdads}
\end{figure*}
This section presents the performance comparison of D$^3$ with its counterparts. 
In addition to executing a DNN model solely on a device node, an edge node, or a cloud node, 
we choose two state-of-the-art precision lossless DNN offloading systems 
Neurosurgeon~\cite{kang2017neurosurgeon} and DADS~\cite{hu2019dynamic} for comparison. 
We mainly examine the effectiveness of HPA and VSM from two performance metrics: 
end-to-end inference speedup and per-image communication overhead. 
To investigate the performance of D$^3$ in various network settings, 
we test the performance metrics under different network conditions shown in 
TABLE~\ref{tab:networkconditions}, which shows the average uplink rate. 
Particularly, we consider the following scenarios.
The communication link between the device node and the edge node is a Wi-Fi running at~5~GHz 
(Gigabit Ethernet 802.3, Wi-Fi 5 IEEE 802.11ac). 
Both the device node and the edge node connect to the the cloud node 
with the same type of communication link, which is Wi-Fi, 4G, or 5G. 
Besides, when the edge node employs optical network to bridge to the cloud node, 
the device node connects to the cloud node via the 5~GHz Wi-Fi. 
In our experiment, we alter the type of communication link between the LAN 
and the cloud node (i.e., Wi-Fi, 4G, 5G, Optical Network), and examine the performance of our algorithms. 

\subsection{End-to-end Inference Speedup} 
This subsection compares the end-to-end inference speedup of D$^3$ with device-only, edge-only, cloud-only, 
Neurosurgeon, and DADS under different network conditions. 
In the edge-only and the cloud-only method, input data is collected by the device node 
and transmitted to the edge node and cloud node respectively for processing.
The experimental results validate the effectiveness of HPA and VSM. 
\begin{table}[t]
	\centering
	\caption{The average uplink rate (Mbps) between two nodes.}
	\setlength{\extrarowheight}{2pt}
	\begin{tabular}{|| C{6.5 em} |  C{4 em} | C{4 em} | C{4 em} | C{4 em}||}
		\hline
		\textbf{DNNs} & \textbf{Wi-Fi} & \textbf{4G} & \textbf{5G} & \textbf{Optical Network} \\ 
		\hline
		device to edge & 84.95 & N.A. & N.A. & N.A.\\ 
		\hline
		edge to cloud & 31.53 & 13.79 & 22.75 & 50.23\\
		\hline
		device to cloud & 18.75 & 6.12 & 11.64 & N.A.\\
		\hline
	\end{tabular}
    \label{tab:networkconditions}
\end{table}

\begin{figure*}[t]
	\centering
	\begin{minipage}[t]{0.49\textwidth}
		\centering
	    \includegraphics[width=\linewidth, align=t]{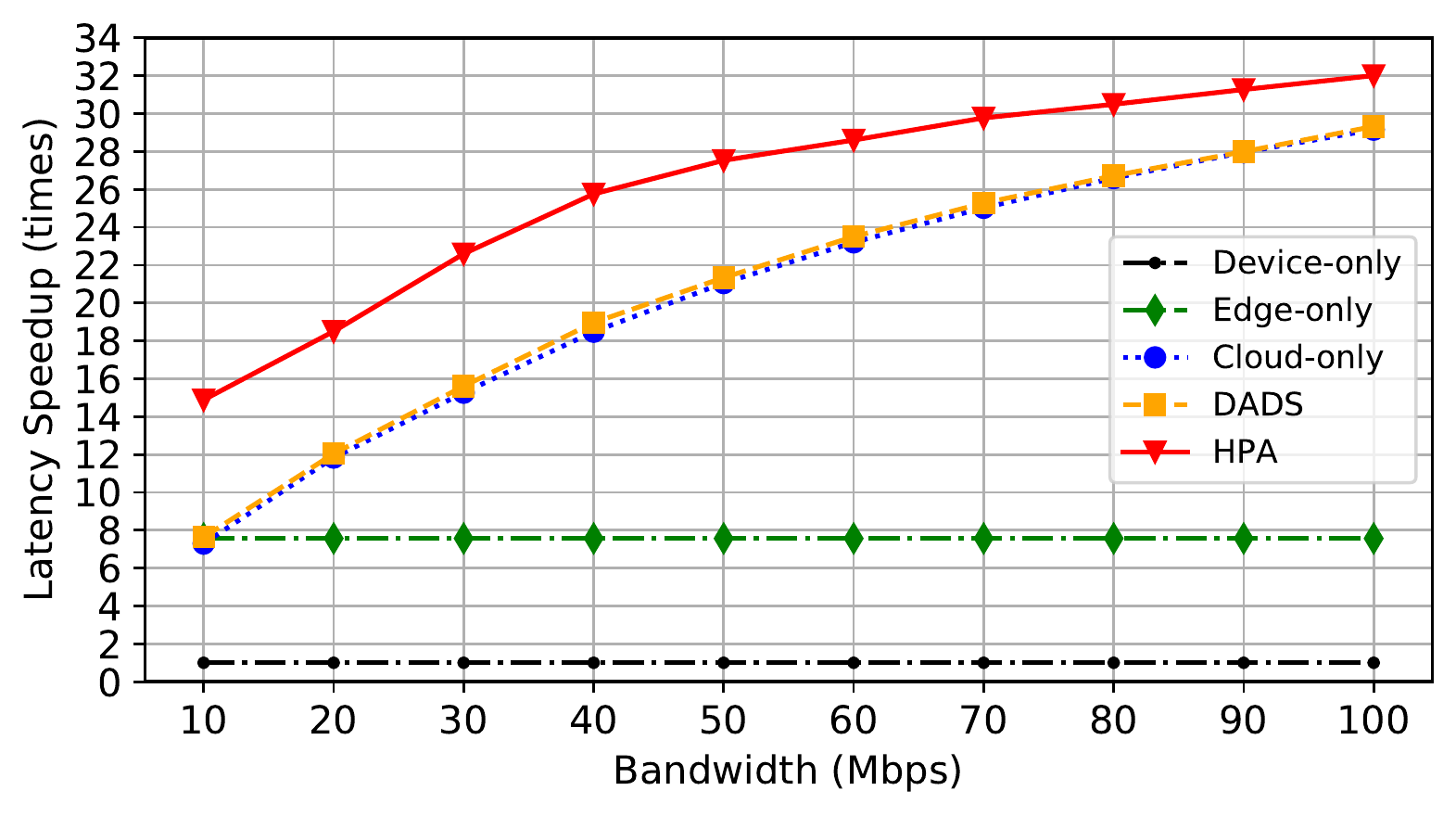}
	    \caption{We measure the latency speedup of Inception-v4 under various bandwidth between the LAN and the cloud node.}
	    \label{plot:bandwidtheffect}
	\end{minipage}\hfill
	\begin{minipage}[t]{0.49\textwidth}
		\centering
	    \includegraphics[width=\linewidth, align=t]{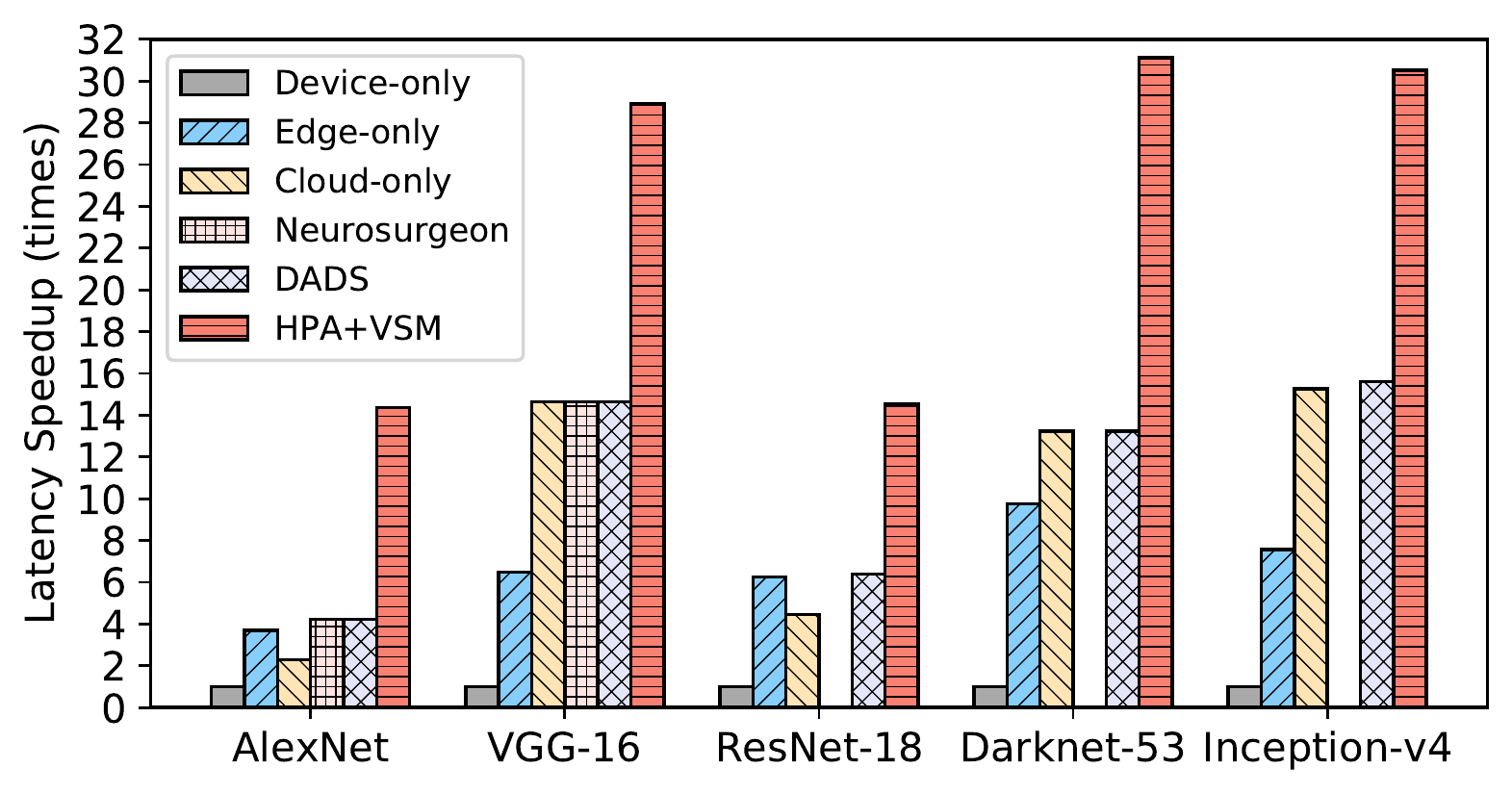}
	    \caption{The latency speedup when applying both HPA and VSM. The device node and the edge nodes connect to the cloud node via Wi-Fi.}
	    \label{plot:wifi_vsm}
	\end{minipage}
\end{figure*}

\begin{figure*}[t]
	\centering
	\begin{subfigure}[t]{0.195\textwidth}
		\centering
		\includegraphics[width=\textwidth]{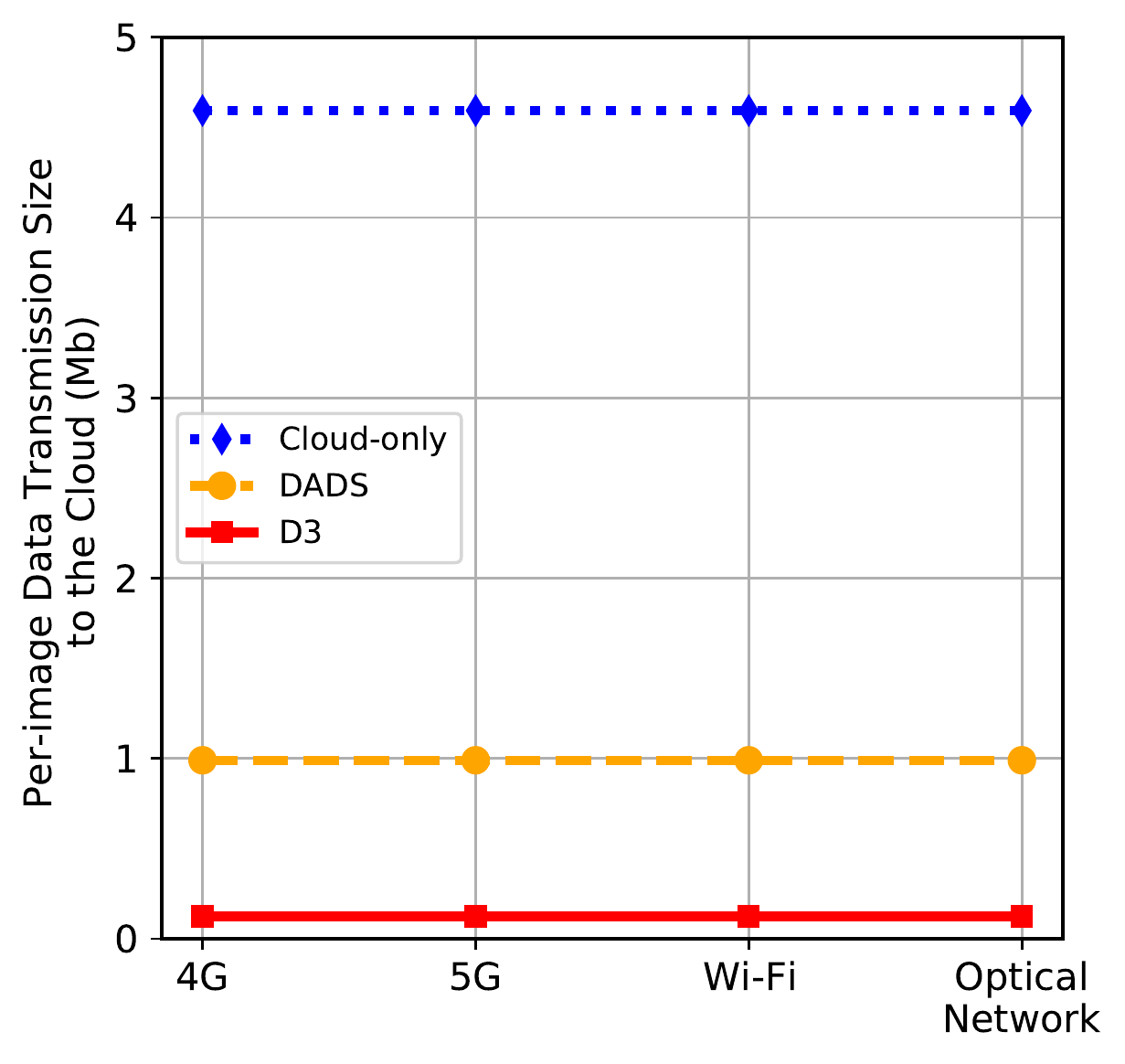}
		\caption{AlexNet}
		\label{plot:dataalex}
	\end{subfigure}
	\hfill
	\begin{subfigure}[t]{0.195\textwidth}
		\centering
		\includegraphics[width=\textwidth]{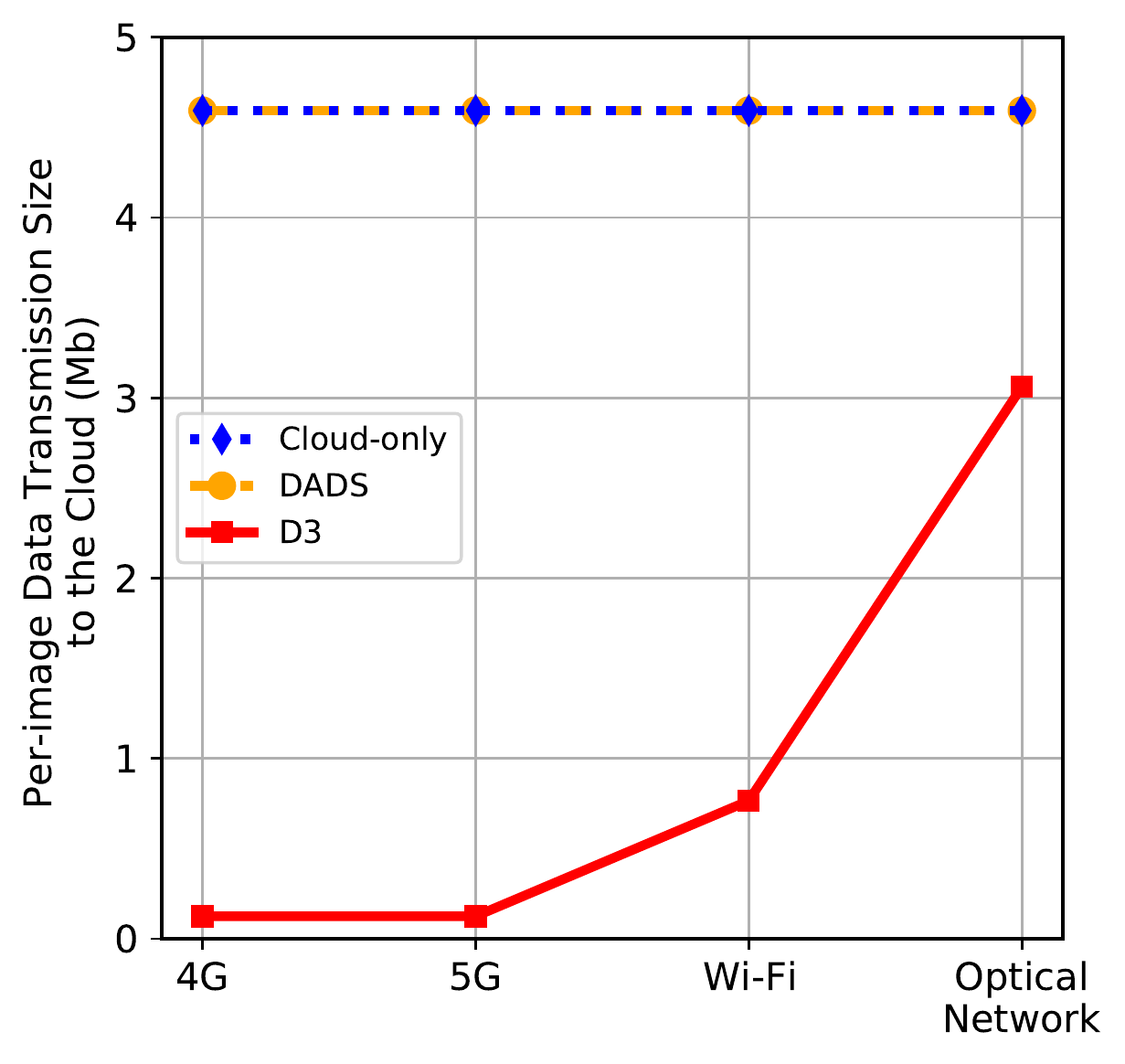}
		\caption{VGG-16}
		\label{plot:datavgg}
	\end{subfigure}
	\hfill
	\begin{subfigure}[t]{0.195\textwidth}
		\centering
		\includegraphics[width=\textwidth]{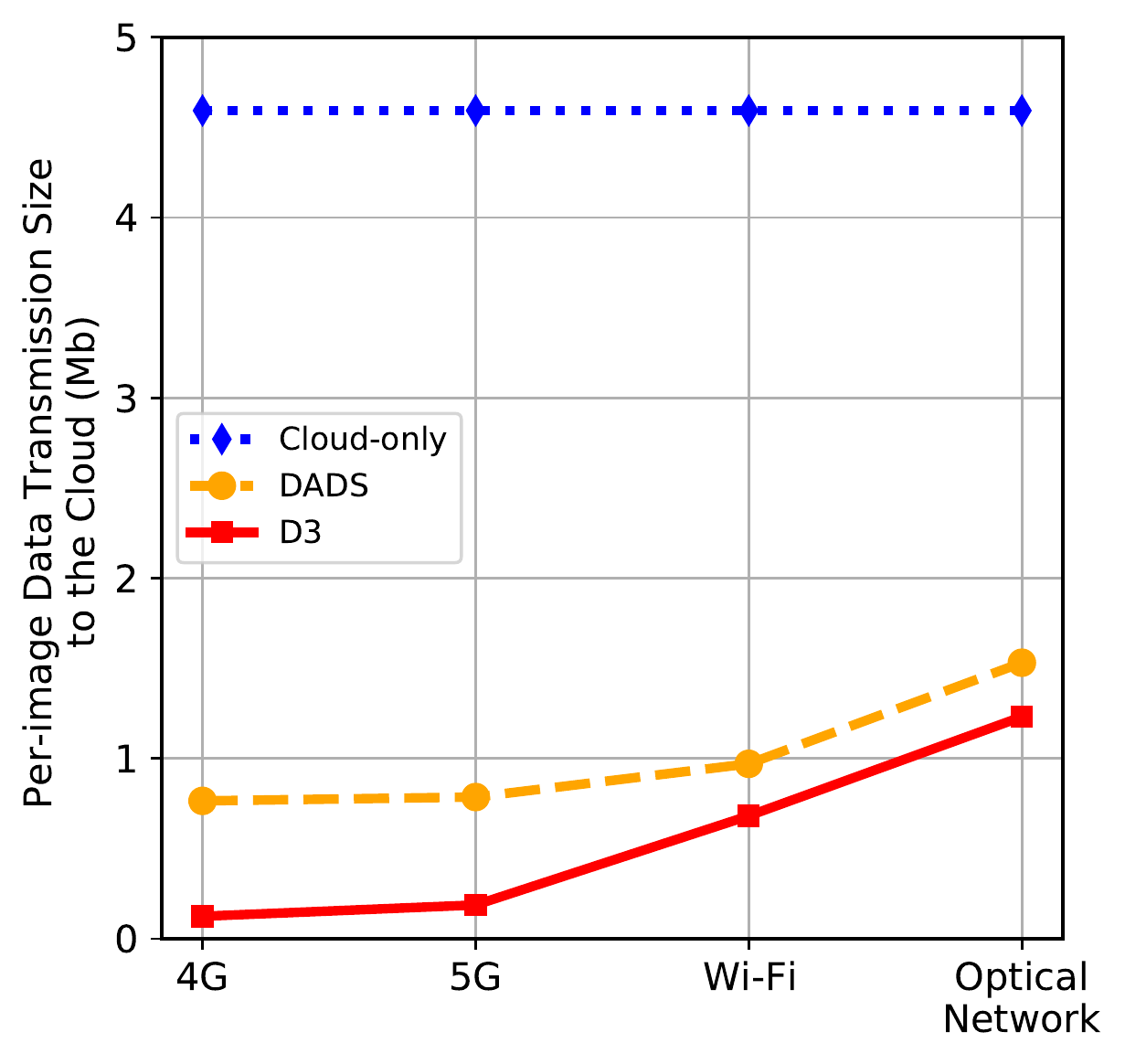}
		\caption{ResNet-18}
		\label{plot:dataresnet}
	\end{subfigure}
	\hfill
	\begin{subfigure}[t]{0.195\textwidth}
		\centering
		\includegraphics[width=\textwidth]{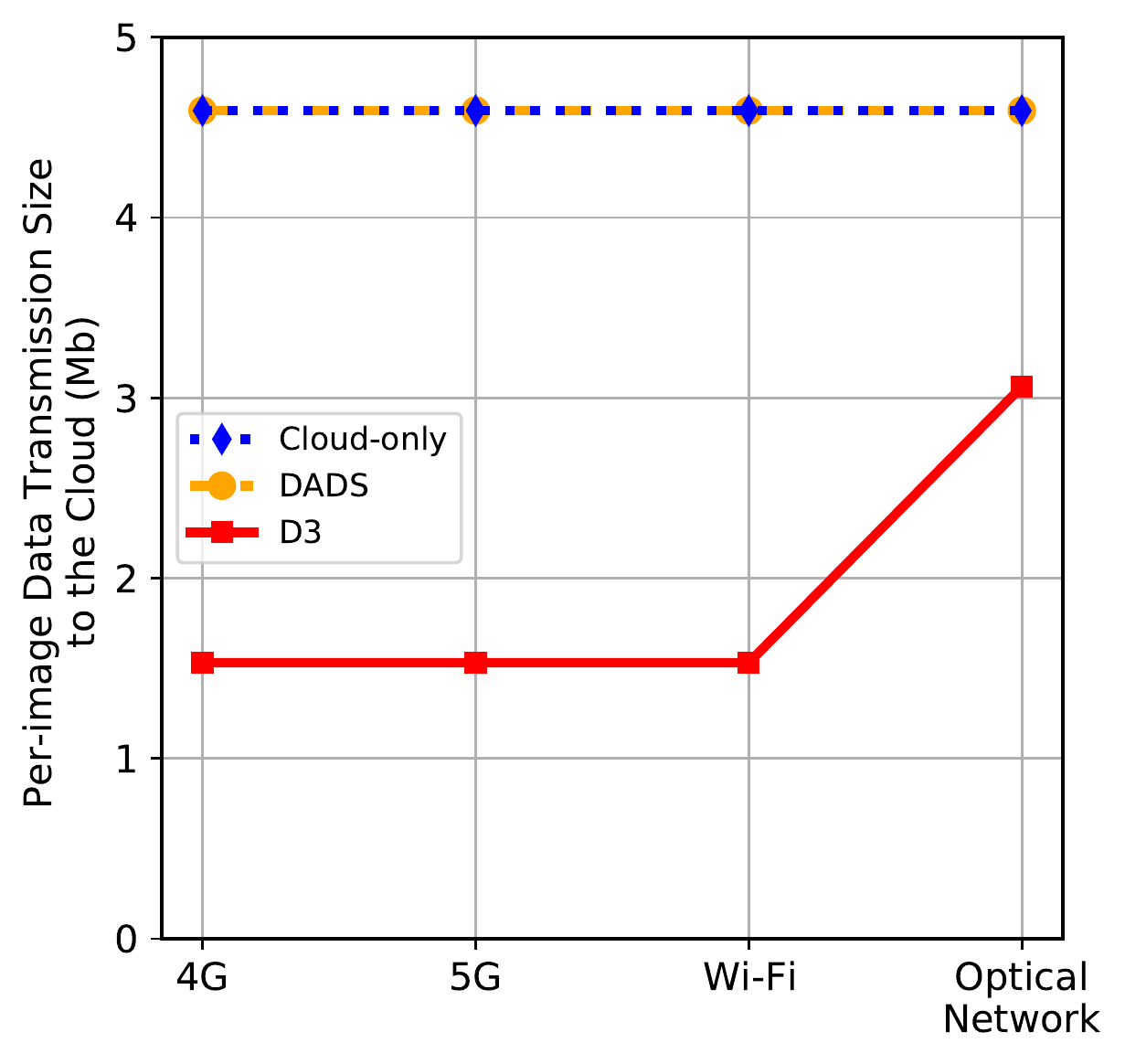}
		\caption{Darknet-53}
		\label{plot:datadarknet}
	\end{subfigure}
	\hfill
	\begin{subfigure}[t]{0.195\textwidth}
		\centering
		\includegraphics[width=\textwidth]{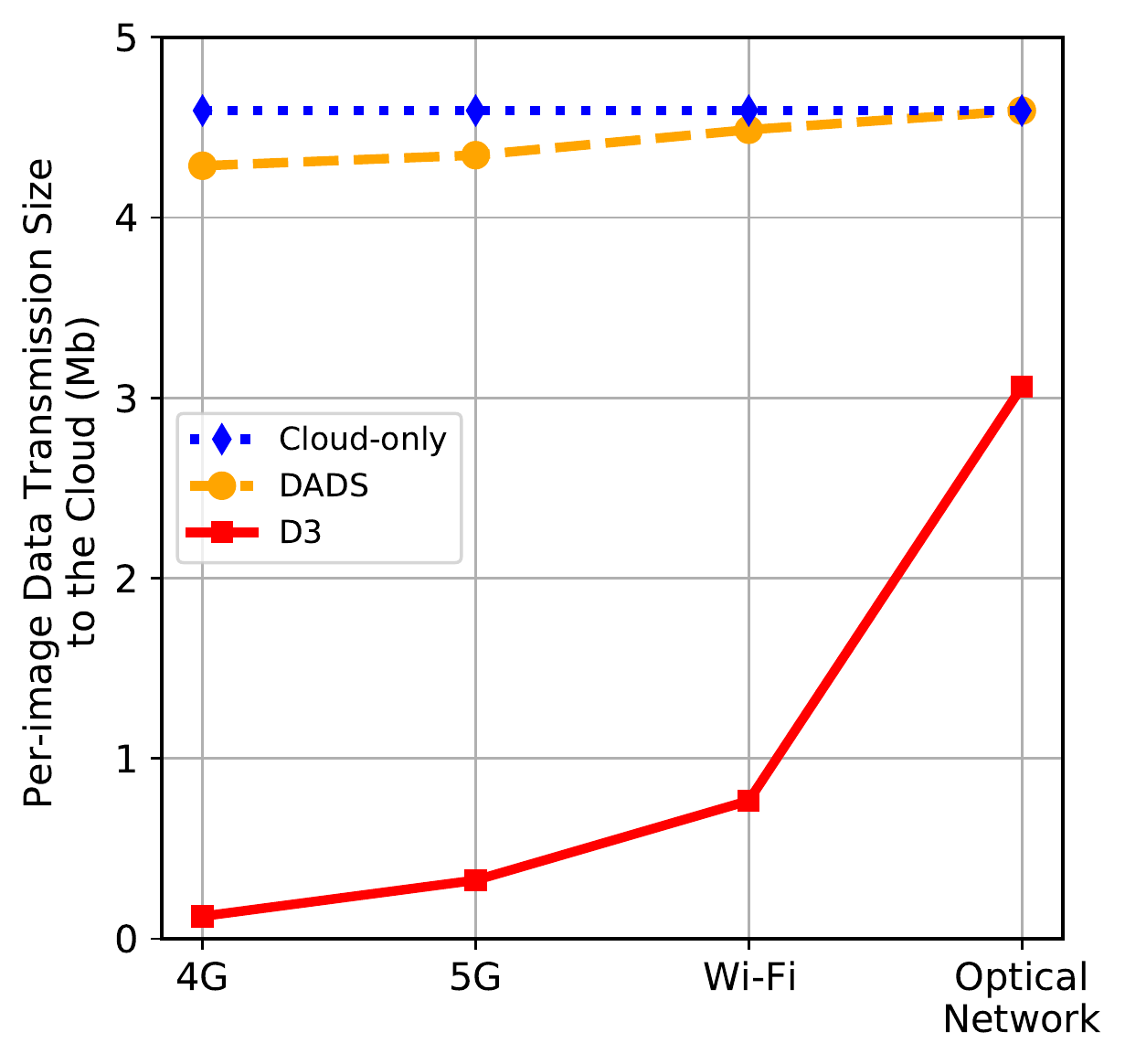}
		\caption{Inception-v4}
		\label{plot:datainception}
	\end{subfigure}
	\caption{Per-image communication overhead comparison among D$^3$, Cloud-only, and DADS~\cite{hu2019dynamic} for different models under different network conditions.}
	\label{plot:data}
\end{figure*} 

Fig.~\ref{plot:localspeedup} demonstrates the end-to-end latency speedup of HPA over device-only approach, 
edge-only approach, and cloud only approach. We set the device-only as the baseline for our comparison. 
The result shows that HPA accelerates the end-to-end latency up to 28.2$\times$, 3.85$\times$, and 5.90$\times$ 
compared with device-only, edge-only, and cloud-only approaches respectively. 
We discover that the device-only is constrained by the computation resources. 
When the model demands more computation power, offloading part of the model 
to the edge node obtains a higher inference speedup.
Meanwhile, the cloud-only method is limited by the low bandwidth between the device node and the cloud node. 
By increasing the bandwidth between the device node and the cloud node, 
the cloud-only method achieves a lower inference latency 
in that the input data transmission time between the device node and the cloud node is smaller. 
Fig.~\ref{plot:localspeedupdads} illustrates the end-to-end latency speedup among HPA, Neurosurgeon, and DADS. 
Since Neurosurgeon can only partition the DNN of chain topology, 
it is not applicable for ResNet-18, Darknet-53, and Inception-v4, which are of DAG topology. 
The experimental results show that HPA outperforms Neurosurgeon 
up to 2.33$\times$ in the chain topology DNNs. 
Compared with DADS, HPA accelerates the inference latency up to 2.97$\times$ in DNNs of DAG topology. 
Notably, we observe that HPA attains more inference speedup when the model gets larger,  
which requires more computation resources. 
To show the impact of network bandwidth variations, 
we apply HPA to Inception-v4 and measure its end-to-end latency speedup 
under different network bandwidth between the LAN and the cloud node. 
From Fig.~\ref{plot:bandwidtheffect}, we monitor that when the network bandwidth 
between the LAN and the cloud node increases, 
HPA tends to offload more DNN layers to the cloud to optimize the end-to-end inference latency. 

We apply VSM to the convolutional layers that are assigned to the edge node under HPA.
We employ four Linux machines with Intel Core i7-8700 CPU and 8 GB system memory as the edge nodes. 
Both the device node and the edge nodes connect to the cloud node via Wi-Fi. 
From Fig.~\ref{plot:wifi_vsm}, we can see the latency speedup when both HPA and VSM are applied. 
The D$^3$ system surpasses the device-only, edge-only, cloud-only, Neurosurgeon, 
and DADS up to 31.13$\times$, 4.46$\times$, 6.28$\times$, 3.4$\times$, and 3.4$\times$ respectively. 
When VSM is deployed, the processing time of convolutional layers at the edge tier
does not shrink to~$1/4$ of the original processing time compared with the HPA-only approach, 
since there are spatial overlaps among the fused tile stacks, 
which in turn leads to computational redundancy. 

Under the condition that both the device node and the edge nodes connect to the cloud node via Wi-Fi, 
HPA improves the end-to-end latency to~1.29$\times$~-~1.8$\times$, 
and HPA+VSM accelerates the latency to~1.96$\times$~-~3.4$\times$, 
where the state-of-the-art is the baseline. 
Overall, the end-to-end latency measurements verify the effectiveness of D$^3$. 

\subsection{Per-image Communication Overhead}
Typically, a cloud node locates at a remote place, which is accessible from a LAN through the Internet. 
A lower data transmission between the LAN and the cloud server 
reduces the data transmission over the network core 
hence mitigating the Internet congestion~\cite{harchol2020making}. 
Transferring the intermediate results of a DNN to the cloud curtails the data transmission between the LAN and the cloud.
We test the per-image communication overhead to the cloud node of D$^3$ and its counterparts
and show the outcomes in Fig.~\ref{plot:data}. 
D$^3$ shrinks the per-image data transmission size on the Internet backbone 
to 27.21\%~-~66.67\% of the cloud-only approach. 
D$^3$ reduces the per-image data transmission size 
on the Internet backbone to 27.21\%~-~80.42\% when DADS is the baseline. 
When the network bandwidth between the LAN and the cloud server becomes larger, 
D$^3$ tends to offload more DNN layers and transmit more intermediate data to the cloud server.

%% file: conclusion.tex
\section{conclusion}
\label{sec:conclusion}
With the increasing computation capability of mobile devices 
and the growing need of accelerating DNN inference, 
there is an expansion demand of performing synergistic 
DNN inference across device, edge, and cloud without precision loss.
In this work, we introduce D$^3$, a system that consists of HPA and VSM. 
HPA partitions a DNN model into three parts 
according to the per-layer processing time and the inter-layer transmission delay of the DNN. 
Besides, VSM further divides the feature maps of DNN layers
assigned to the edge node into multiple fused tile stacks for parallel processing. 
Under our experiment settings, D$^3$ system provides up to 3.4$\times$ end-to-end inference speedup 
on chain topology DNNs and accelerates the end-to-end inference up to 2.35$\times$ 
on DAG topology DNNs compared to the state-of-the-art. 
In addition, the per-image data transmission size reduces to 27.21\% of the state-of-the-art, 
which relieves network congestion and communication cost. 